\documentclass[11pt]{article}
\usepackage[margin=2.7cm]{geometry} 
\usepackage{graphicx} 
\usepackage[colorlinks=false]{hyperref} 
\usepackage{amssymb} 
\usepackage{amsfonts} 
\usepackage{amsmath,amscd} 
\usepackage{amsthm}  
\usepackage{mathtools} 
\usepackage{titlesec} 
\usepackage{csquotes} 
\usepackage[shortlabels]{enumitem} 
\usepackage{subfiles} 
\usepackage{tikz-cd} 
\usepackage{wrapfig} 
\usepackage{tabularx} 
\usepackage{hyperref}
\hypersetup{
     colorlinks   = true,
     citecolor    = green
}

\usepackage{nomencl}

\makenomenclature

\usepackage{afterpage}

\setlist[itemize]{itemsep=0mm} 
\numberwithin{equation}{section} 
\titleformat*{\section}{\Large \scshape\center} 
\titleformat*{\subsection}{\fontsize{14}{14} \sffamily} 

\theoremstyle{plain}
\newtheorem{theorem}{Theorem}[section]
\newtheorem*{theorem*}{Theorem} 
\newtheorem*{proposition*}{Proposition} 
\newtheorem{lemma}[theorem]{Lemma}
\newtheorem{proposition}[theorem]{Proposition}
\newtheorem{corollary}[theorem]{Corollary}

\theoremstyle{definition}
\newtheorem{definition}[theorem]{Definition}
\newtheorem{example}[theorem]{Example}

\newcommand{\Addresses}{{
  \bigskip
  \footnotesize

  \textsc{Department of Mathematical Sciences, Norwegian University of Science and Technology,\\ 7491 Trondheim, Norway.}\par\nopagebreak
  \textit{E-mail addresses}: \texttt{eirik.berge@ntnu.no}, \texttt{stine.m.berge@ntnu.no}, and \texttt{franz.luef@ntnu.no}
}}

\theoremstyle{remark}
\newtheorem*{remark}{Remark}
\newtheorem*{remarks}{Remarks}

\begin{document}
\pagenumbering{gobble}
\title{\Huge{The Affine Wigner Distribution}}
\author{Eirik Berge, Stine Marie Berge, and Franz Luef}
\date{}
\maketitle

\pagenumbering{arabic}

\begin{abstract}
    We examine the affine Wigner distribution from a quantization perspective with an emphasis on the underlying group structure. One of our main results expresses the scalogram as (affine) convolution of affine Wigner distributions. We strive to unite the literature on affine Wigner distributions and we provide the connection to the Mellin transform in a rigorous manner. Moreover, we present an affine ambiguity function and show how this can be used to illuminate properties of the affine Wigner distribution. In contrast with the usual Wigner distribution, we demonstrate that the affine Wigner distribution is never an analytic function. \par
    Our approach naturally leads to several applications, one of which is an approximation problem for the affine Wigner distribution. We show that the deviation for a symbol to be an affine Wigner distribution can be expressed purely in terms of intrinsic operator-related properties of the symbol. Finally, we present an \textit{affine positivity conjecture} regarding the non-negativity of the affine Wigner distribution.
\end{abstract}

\section{Introduction}

The most studied quadratic time-frequency representation is the \textit{Wigner distribution} defined by
\begin{equation}
\label{Wigner_distribution}
    W_{f}(x,\omega) := \int_{\mathbb{R}^d}f\left(x + \frac{t}{2}\right)\overline{f\left(x - \frac{t}{2}\right)}e^{-2\pi i \omega t} \, dt, \quad (x,\omega) \in \mathbb{R}^{2d}.
\end{equation}
Originally invented by Wigner in \cite{wigner1997quantum} almost a century ago, the Wigner distribution is essential in quantum mechanics as it gives the expectation values for Weyl quantization of symbols \cite{ercolessi2007wigner}. In recent decades, the Wigner distribution has found many applications in time-frequency analysis \cite[Chapter 4]{grochenig2013foundations} due to its connections with the short-time Fourier transform $V_{g}f$ defined precisely in \eqref{STFT}. One of the more surprising connections is the convolution relation 
\begin{equation}
\label{wigner_convolution}
    |V_{g}f(x,\omega)|^2 = W_{P(g)} * W_{f}(x,\omega),
\end{equation} 
where $P$ is the reflection operator $P(g)(x) := g(-x)$. The function $\mathrm{SPEC}_{g}f := |V_{g}f(x,\omega)|^2$ is called the \textit{spectrogram} of $f$ with window $g$. The spectogram is an important tool for analyzing time-frequency content and has been used extensively in the engineering literature since its introduction. \par 
Parallel to the theory of time-frequency analysis is the time-scale (or wavelet) paradigm. Although there have been many attempts at finding a suitable Wigner distribution in the time-scale setting, there is no general consensus in the literature.
We will motivate a particular choice of a time-scale Wigner distribution $W_{\mathrm{Aff}}^\psi$ given by \[W_{\mathrm{Aff}}^\psi(x,a)
:= \int_{-\infty}^{\infty}\psi\left(\frac{aue^u}{e^u - 1}\right)\overline{\psi\left(\frac{au}{e^{u} - 1}\right)}e^{-2\pi i x u} \, du, \quad (x,a) \in \mathbb{R} \times \mathbb{R}_{+}.\] The function $W_{\mathrm{Aff}}^\psi$ is called the \textit{affine Wigner distribution} due to its relation to the \textit{affine group} $\mathrm{Aff}$. It was derived through a quantization procedure in \cite{gayral2007fourier}. The authors showed that the affine Wigner distribution satisfies $W_{\mathrm{Aff}}^\psi \in L_{r}^{2}(\mathrm{Aff})$ for every $\psi \in L^{2}(\mathbb{R}_+, a^{-1}\,da)$, where $L_{r}^{2}(\mathrm{Aff})$ denotes all measurable functions on the upper half-plane $\mathbb{R} \, \times \, \mathbb{R}_{+}$ that are square integrable with respect to the measure $a^{-1} \, da \, dx$. \par 
The affine Wigner distribution $W_{\mathrm{Aff}}^\psi$ has appeared in the literature several times throughout the years; as a particular Bertrand distributions in \cite{papandreou1998quadratic}, and as a tool for studying the quantum mechanics of the Morse potential in \cite{molnar2001coherent}. The basic properties of the affine Wigner distribution will be developed in a rigorous manner to fill gaps in the literature. In particular, for all sufficiently nice $\psi \in L^{2}(\mathbb{R}_+, a^{-1} \, da)$ we have the \textit{marginal properties} 
\begin{equation*}
    \int_{-\infty}^{\infty} W_{\mathrm{Aff}}^{\psi}(x,a) \, dx = |\psi(a)|^2,
\end{equation*} 
\begin{equation}
\label{Mellin_margin_in_introduction}
    \int_{0}^{\infty} W_{\mathrm{Aff}}^{\psi}(x,a) \, \frac{da}{a} = |\mathcal{M}(\psi)(x)|^2.
\end{equation}
The symbol $\mathcal{M}(\psi)(x)$ denotes the \textit{Mellin transform} of $\psi \in L^{2}(\mathbb{R}_+, a^{-1} \, da)$ at the point $x \in \mathbb{R}$ given by \[\mathcal{M}(\psi)(x) = \mathcal{M}_{a}(\psi)(x) := \int_0^\infty \psi(a)a^{-2\pi i x}\, \frac{da}{a}.\]

The first significant contribution is to develop a connection between the affine Wigner distribution and the \textit{scalogram} defined by\[\mathrm{SCAL}_{g}f(x,a) := |\mathcal{W}_{g}f(x,a)|^2, \quad (x,a) \in \mathbb{R} \times \mathbb{R}_{+},\] where $\mathcal{W}_{g}f$ denotes the continuous wavelet transform of $f$ with respect to $g$ defined precisely in \eqref{continuous_wavelet_transform}. By comparing with \eqref{wigner_convolution} in the time-frequency setting, one would expect a simple convolution relation to hold. However, as the group underlying the symmetries in the time-scale case is the affine group we obtain the following result.

\begin{theorem*}
Let $f$ and $g$ be square integrable functions on the real line and assume their Fourier transforms $\phi := \widehat{f}$ and $\psi := \widehat{g}$ are supported in $\mathbb{R}_{+}$ and are in $L^{2}(\mathbb{R}_{+},a^{-1} \, da)$. Then the scalogram of $f$ with window $g$ is given by the affine convolution \[\mathrm{SCAL}_{g}{f}(x,a) = \left(I\left(W_{\mathrm{Aff}}^{\psi}\right) *_{\mathrm{Aff}} \Delta W_{\mathrm{Aff}}^{\phi}\right)\left(\frac{x}{a}, \frac{1}{a}\right), \quad (x,a) \in \mathbb{R} \times \mathbb{R}_{+},\]
where $\Delta$ and $I$ denote the modular function and the involution on the affine group, respectively. 
\end{theorem*}

The affine group $\mathrm{Aff}$ and constructions on it will be thoroughly described in the sections to come. We will use the affine group time and time again to shed light on various results we derive. In particular, it will be clear that the measure used in the marginal property \eqref{Mellin_margin_in_introduction} is the correct one. \par 
We introduce an \textit{affine ambiguity function} $A_{\mathrm{Aff}}^{\psi}$ for $\psi \in L^{2}(\mathbb{R}_+, a^{-1} \, da)$ given by 
\[A_{\mathrm{Aff}}^{\psi}(x,a) := \int_{0}^{\infty} \psi\left(r \sqrt{a}\right)\overline{\psi\left(\frac{r}{\sqrt{a}}\right)}r^{-2\pi i x} \, \frac{dr}{r}, \quad (x,a) \in \mathbb{R} \times \mathbb{R}_{+}.\]
The affine ambiguity function is intimately related to the radar ambiguity function in time-frequency analysis \cite[Chapter 4.2]{grochenig2013foundations}. We will show that the affine Wigner distribution and the affine ambiguity function are related through the Mellin transform by 
\begin{equation}
\label{lemma_Mellin_introduction}
    W_{\mathrm{Aff}}^{\psi}(x,a) = \mathcal{M}_{y}^{-1} \otimes \mathcal{M}_{b}\left[\left(\frac{\sqrt{b}\log(b)}{b - 1}\right)^{2\pi i y}  A_{\mathrm{Aff}}^{\psi}(y,b)\right](x,a).
\end{equation}
The relation \eqref{lemma_Mellin_introduction} is used in the proof of Proposition \ref{Wigner_Schwartz_functions} to show that the affine Wigner distribution preserves Schwartz functions. \par It turns out that affine Wigner distributions are never analytic functions on the upper half-plane. This is a consequence of the non-existence of analytic functions in the space $L_{r}^{2}(\mathrm{Aff})$. However, the space $L_{r}^{2}(\mathrm{Aff})$ can be completely decomposed into \textquote{almost analytic} functions as the following result shows.

\begin{proposition*}
We have the orthogonal decomposition 
\begin{equation}
\label{orthogonal_decomposition_equation_introduction}
    L_{r}^{2}(\mathrm{Aff}) = \bigoplus_{n = 2}^{\infty}\mathcal{A}^{n}(\mathrm{Aff}) \oplus \mathcal{A}^{\perp,n}(\mathrm{Aff}),
\end{equation}
where $\mathcal{A}^{n}(\mathrm{Aff})$ and $\mathcal{A}^{\perp,n}(\mathrm{Aff})$ denote the spaces of pure poly-analytic and pure anti-poly-analytic functions of order $n$, respectively. In particular, there are no analytic or anti-analytic functions on the form $W_{\mathrm{Aff}}^{\psi}$ for $\psi \in L^{2}(\mathbb{R}_{+}, a^{-1} \, da)$. 
\end{proposition*}

As an application to the theory developed we consider the approximation problem of understanding, for a given $f \in L_{r}^{2}(\mathrm{Aff})$, the quantity 
\begin{equation}
\label{approximation_problem_introduction}
    \inf_{\psi \in L^{2}(\mathbb{R}_{+}, a^{-1} \, da)}\left\|f - W_{\mathrm{Aff}}^\psi\right\|_{L_{r}^{2}(\mathrm{Aff})}.
\end{equation}
Notice that \eqref{approximation_problem_introduction} measures how far $f$ is from being an affine Wigner distribution. The analogous problem in time-frequency analysis has been recently studied in \cite{ben2018wigner}. For each symbol $f \in L_{r}^{2}(\mathrm{Aff})$ there is a Hilbert-Schmidt operator $A_f$ on $L^{2}(\mathbb{R}_{+}, a^{-1} \, da)$ that is weakly defined by the relation
\begin{equation}
\label{quantization_relation_introduction}
    \left\langle A_{f}\psi,\phi \right\rangle_{L^{2}(\mathbb{R}_{+},a^{-1} \, da)} = \left\langle f,W_{\mathrm{Aff}}^{\phi,\psi}\right\rangle_{L_{r}^{2}(\mathrm{Aff})}, \quad \psi,\phi \in L^{2}(\mathbb{R}_{+},a^{-1} \, da).
\end{equation}
The following result shows that the quantity \eqref{approximation_problem_introduction} is linked to how much $A_f$ deviates from being a rank-one operator.

\begin{theorem*}
Let $f \in L_{r}^{2}(\mathrm{Aff})$ be real valued. Then (under a mild eigenvalue assumption on $A_f$) we have that \[\inf_{\psi \in L^{2}(\mathbb{R}_{+}, a^{-1} \, da)}\left\|f - W_{\mathrm{Aff}}^\psi\right\|_{L_{r}^{2}(\mathrm{Aff})} = \sqrt{\|A_f\|_{\mathcal{HS}}^2 - \|A_{f}\|_{op}^2},\] where $\|\cdot\|_{\mathcal{HS}}$ and $\|\cdot\|_{op}$ are the Hilbert-Schmidt norm and operator norm, respectively. Moreover, the precise number of distinct minimizers can be deduced from the spectrum of $A_{f}$. 
\end{theorem*}

The structure of the paper is as follows: In Section \ref{sec: Preliminaries} we outline nessesary definitions and briefly review the affine group as it will be central for many of the results we develop. In Section \ref{sec: Baisc_Properties} we derive basic properties of the affine Wigner distribution.
We devote Section \ref{sec: Alternative_Descriptions} to uniting the literature and pointing out how the affine Wigner distribution can be derived by emphasizing symmetry. The convolution relation between the affine Wigner distribution and the scalogram will be proved in Section \ref{sec: Convolution_Relation}. In Section \ref{sec: Affine_ambiguity_function_section} we define the affine ambiguity function and show how this allows us to extend the affine quantization \eqref{quantization_relation_introduction} to the distributional setting. We prove the decomposition \eqref{orthogonal_decomposition_equation_introduction} of $L_{r}^{2}(\mathrm{Aff})$ in Section \ref{sec: Polyanalytic_Decomposition} and show how the Laguerre polynomials are a useful tool in our setting. In addition to the approximation problem described above, we show in Section \ref{sec: Applications} how basic questions regarding operators on $\mathbb{R}_+$ can be answered with our framework. Finally, we discuss the affine Grossmann-Royer operator and the affine positivity conjecture in Section \ref{sec: Open_Questions}. The authors are grateful for helpful suggestions from Eirik Skrettingland and Lu\'{i}s Daniel Abreu.

\section{Preliminaries}
\label{sec: Preliminaries}

The notation $\mathcal{S}(\mathbb{R}^d)$ will be used for the Schwartz space of rapidly decaying smooth functions on $\mathbb{R}^d$. Its dual space of tempered distributions is denoted by $\mathcal{S}'(\mathbb{R}^d)$. The Fourier transform of a function $f \in L^{2}(\mathbb{R}^d)$ will be denoted by \[\mathcal{F}f(\omega) = \hat{f}(\omega) := \int_{\mathbb{R}^d} f(x) e^{-2\pi i x \omega} \, dx, \quad \omega \in \mathbb{R}^d.\] \par 
We will frequently consider the space $L^{2}(\mathbb{R}_{+}) := L^{2}(\mathbb{R}_+, a^{-1} \, da)$ consisting of measurable functions $f:\mathbb{R}_+ \to \mathbb{C}$ such that \[\int_{0}^{\infty}|f(a)|^2 \, \frac{da}{a} < \infty.\] This notation is consistent with our group theoretical approach as $a^{-1} \, da$ is the Haar measure on $\mathbb{R}_+$. If we want to consider the usual Euclidean measure on $\mathbb{R}_+$, we will explicitly write $L^{2}(\mathbb{R}_+, dx)$ to avoid any confusion. We use the notation $\mathcal{S}(\mathbb{R}_+)$ for the smooth functions $\psi:\mathbb{R}_+ \to \mathbb{C}$ such that $\Phi(x) := \psi(e^{x}) \in \mathcal{S}(\mathbb{R})$. \par 
The reader is referred to Appendix \ref{Appendix_1} for notation and basic properties regarding Schatten class operators $\mathcal{S}_{p}(\mathcal{H})$ on a separable Hilbert space $\mathcal{H}$ for $1 \leq p < \infty$. In particular, the Hilbert-Schmidt operators on $\mathcal{H}$ will be denoted by $\mathcal{S}_{2}(\mathcal{H})$.

\subsection{The Classical Wigner Distribution and the Heisenberg Group}
\label{sec: subsection_Classical_Wigner}

We begin by recalling basic definitions from time-frequency analysis and their connection with the Heisenberg group. The \textit{cross-Wigner transform} $W(f,g)$ of $f, g \in L^{2}(\mathbb{R}^d)$ is defined to be 
\begin{equation*}
    W(f,g)(x,\omega) := \int_{\mathbb{R}^d}f\left(x + \frac{t}{2}\right)\overline{g\left(x - \frac{t}{2}\right)}e^{-2\pi i \omega t} \, dt, \quad (x,\omega) \in \mathbb{R}^{2d}.
\end{equation*}
Notice that the Wigner distribution $W_{f}$ given in \eqref{Wigner_distribution} is precisely the diagonal term $W(f,f)$. The cross-Wigner transform satisfies the orthogonality property 
\begin{equation}
\label{usual_Wigner_orthogonality}
    \langle W(f_1,g_1), W(f_2,g_2) \rangle_{L^{2}(\mathbb{R}^{2d})} = \langle f_1, f_2 \rangle_{L^{2}(\mathbb{R}^d)}\overline{\langle g_1, g_2 \rangle}_{L^{2}(\mathbb{R}^d)}.
\end{equation} \par 
A key feature of the Wigner distribution is its connection with the Weyl calculus: For a symbol $\sigma \in \mathcal{S}'(\mathbb{R}^{2d})$, the \textit{Weyl (pseudo-differential) operator} $L_{\sigma}$ corresponding to the symbol $\sigma$ is the operator
\begin{equation}
\label{pseudo_equation}
    L_{\sigma}f := \int_{\mathbb{R}^{2d}} e^{-\pi i \xi u}\hat{\sigma}(\xi, u) T_{-u} M_{\xi} f\, du \, d\xi.
\end{equation}
The operators $T_{-u}$ and $M_{\xi}$ in \eqref{pseudo_equation} are respectively the \textit{time-shift operator} and the \textit{frequency-shift operator} defined by \begin{equation*}
    T_{x}f(t) := f(t - x), \qquad M_{\omega}f(t) := e^{2 \pi i \omega t}f(t), \quad x, \omega, t \in \mathbb{R}^d.
\end{equation*}
The association $\sigma \mapsto L_{\sigma}$ is called the \textit{Weyl transform} and the operator $L_{\sigma}$ maps $\mathcal{S}(\mathbb{R}^d)$ into $\mathcal{S}'(\mathbb{R}^d)$ by \cite[Lemma 14.3.1]{grochenig2013foundations}. Moreover, the Weyl transform is a one-to-one correspondence between square integrable symbols $\sigma \in L^{2}(\mathbb{R}^{2d})$ and Hilbert-Schmidt operators $L_{\sigma} \in \mathcal{S}_{2}\left(L^{2}(\mathbb{R}^d)\right)$ by a classical result of Poole \cite[Proposition V.1]{pool1966mathematical}. The connection between the Weyl calculus and the cross-Wigner transform is the relation
\begin{equation*}
    \langle L_{\sigma}f,g \rangle_{L^{2}(\mathbb{R}^{d})} = \langle \sigma, W(g,f) \rangle_{L^{2}(\mathbb{R}^{2d})},
\end{equation*} 
for $\sigma \in L^{2}(\mathbb{R}^{2d})$ and $f,g \in L^{2}(\mathbb{R}^{d})$.
Since the Weyl transform is a quantization procedure, one can think of the inverse transformation $L_{\sigma} \mapsto \sigma$ for $L_{\sigma} \in \mathcal{S}_{2}\left(L^{2}(\mathbb{R}^d)\right)$ as \textit{dequantization}. In this terminology, the Wigner distribution $W_{f}$ for $f \in L^{2}(\mathbb{R}^d)$ is the dequantization of the rank-one operator 
\begin{equation}
\label{Wigner_as_dequantization}
    L_{W_f}g := \langle g, f \rangle f, \qquad g \in L^{2}(\mathbb{R}^d).
\end{equation}
The reader should consult \cite[Chapter 13]{hall2013quantum} and \cite[Chapter 4]{de2017wigner} for more details about the Weyl transform from a quantum mechanical perspective.  \par
Central to time-frequency analysis and its transforms is the representation theory of the Heisenberg group $\mathbb{H}^{2d + 1}$. We can realize the Heisenberg group $\mathbb{H}^{2d + 1}$ as $\mathbb{R}^d \times \mathbb{R}^d \times \mathbb{R}$ where the multiplication between two elements is given by
\begin{equation*}
    (x,\omega,t) \cdot_{\mathbb{H}^{2d + 1}} (x',\omega',t') := \left(x + x', \omega + \omega', t + t' + \frac{1}{2}\left(x' \omega - x \omega' \right)\right).
\end{equation*}
The Heisenberg group $\mathbb{H}^{2d+1}$ acts on the space $L^{2}(\mathbb{R}^d)$ by 
\begin{equation*}
    \pi(x,\omega,t)(f) := e^{2\pi i t}e^{\pi i x \omega}T_{x} M_{\omega} f.
\end{equation*}
We refer to $\pi$ as the \textit{Schr\"{o}dinger representation} and it is both unitary and irreducible. \par A short computation reveals that the matrix coefficients of the Schr\"{o}dinger representation is given by \[\langle f, \pi(x,\omega,t)g \rangle_{L^{2}(\mathbb{R}^d)} = e^{\pi i (x\omega - 2t)}V_{g}f(x,\omega),\] where $V_{g}f$ is the \textit{short-time Fourier transform (STFT)} given by 
\begin{equation}
\label{STFT}
    V_{g}f(x,\omega) := \langle f, M_{\omega} T_{x}g \rangle_{L^{2}(\mathbb{R}^d)} = \int_{\mathbb{R}^d}f(t)\overline{g(t-x)}e^{-2\pi i \omega t} \, dt.
\end{equation}
We have from \cite[Lemma 4.3.1]{grochenig2013foundations} that the cross-Wigner transform and the STFT is related by the formula \[W(f,g)(x,\omega) = 2^de^{4\pi ix\omega}V_{P(g)}f(2x,2\omega),\] where $P(g)(x) := g(-x)$. Thus the matrix coefficients of the Schr\"odinger representation is related to the Wigner distribution by the formula \[W_{f}(x,\omega) = 2^{d}e^{2\pi i t}\langle f, \pi(2x,2\omega,t)P(f) \rangle_{L^{2}(\mathbb{R}^d)}.\] Moreover, the Stone-von Neumann theorem \cite[Theorem 9.3.1]{grochenig2013foundations} reinforces the importance of the Schr\"odinger representation by declaring that the only infinite-dimensional, irreducible, unitary representations of the Heisenberg group are appropriate dilations of the Schr\"odinger representation. 

\subsection{Wavelet Transforms and the Affine Group}

The two main operators in time-scale analysis are the time-shift operator $T_{x}$ and the \textit{dilation operator} $D_{a}$ given by 
\begin{equation}
\label{dilation_operator}
    D_{a}f(x) := \frac{1}{\sqrt{a}}f\left(\frac{x}{a}\right),
\end{equation}
for $a > 0$ and $f \in L^{2}(\mathbb{R})$. In time-scale (or wavelet) analysis, the affine group takes the role that the Heisenberg group has in time-frequency analysis. The \textit{affine group} $\mathrm{Aff} := (\mathbb{R} \times \mathbb{R}_{+}, \cdot_{\mathrm{Aff}})$ has the group operation \[(x,a) \cdot_{\mathrm{Aff}} (y,b) := (x + ay,ab), \quad (x,a),(y,b) \in \mathrm{Aff},\] motivated by the composition rule \[(T_x D_a)(T_y D_b) = T_x T_{ay} D_a D_b = T_{x + ay}D_{ab}.\] \par 
We can represent the affine group $\mathrm{Aff}$ and its Lie algebra $\mathfrak{aff}$ as \[\mathrm{Aff} = \left\{\begin{pmatrix} a & x \\ 0 & 1\end{pmatrix}\Big| \, a> 0, \, x \in \mathbb{R}\right\}, \quad \mathfrak{aff} = \left\{\begin{pmatrix} u & v \\ 0 & 0\end{pmatrix}\Big| \, u,v \in \mathbb{R}\right\}.\] Essential for computations is the fact that the exponential map $\exp:\mathfrak{aff} \to \mathrm{Aff}$ given by \[\exp\begin{pmatrix} u & v \\ 0 & 0 \end{pmatrix} = \begin{pmatrix} e^{u} & \frac{v(e^{u} - 1)}{u} \\ 0 & 1\end{pmatrix}\] is a global diffeomorpism. The left Haar measure on $\mathrm{Aff}$ is given by $a^{-2} \, da \, dx$, while the right Haar measure is $a^{-1} \, da \, dx$. We will use the notation $L_{r}^{2}(\mathrm{Aff})$ and $L_{l}^{2}(\mathrm{Aff})$ to indicate if we are using the right or left Haar measure, respectively. The left and right Haar measures on $\mathrm{Aff}$ can be written in the coordinates induced by the exponential map as \[\frac{da \, dx}{a^2} = \frac{du \, dv}{\lambda(u)}, \qquad \frac{da \, dx}{a} = \frac{du \, dv}{\lambda(-u)},\] where the function $\lambda$ is given by \begin{equation}
\label{lambda_function}
    \lambda(u) := \frac{ue^{u}}{e^{u} - 1} = \frac{ue^{\frac{u}{2}}}{2\sinh(\frac{u}{2})}.
\end{equation}
\par 
A natural way the affine group can act on $L^{2}(\mathbb{R})$ is by translations and dilations, namely as 
\begin{equation}
\label{space_side_representation}
    f \longmapsto T_{x}D_{a}f, \qquad f \in L^{2}(\mathbb{R}).
\end{equation} This is a unitary representation, although it is not irreducible. The matrix coefficients of this representation are given by 
\begin{equation}
\label{continuous_wavelet_transform}
    \mathcal{W}_{g}f(x,a) := \langle f,T_{x}D_{a}g \rangle_{L^{2}(\mathbb{R})} = \frac{1}{\sqrt{a}}\int_{-\infty}^{\infty}f(y)\overline{g\left(\frac{y-x}{a}\right)} \, dy.
\end{equation}
One typically refer to the map $(x,a) \mapsto \mathcal{W}_{g}f(x,a)$ as the \textit{(continuous) wavelet transform} of $f$ with respect to $g$. The continuous wavelet transform is analogous to the STFT and incorporates the possibility of observing $f$ at different scales through $g$. Moreover, the magnifying aspect coming from the change of scales can characterize local regularity through decay properties of the wavelet transform, see \cite[Theorem 2.9.2]{daubechies1992ten}.

\subsection{A Quantization Approach to the Affine Wigner Distribution}
\label{section: Kirillov}
We will briefly outline a procedure described in \cite{gayral2007fourier} to determine the affine Wigner distribution. The theory is based on Kirillov's theory of coadjoint orbits and we refer further explanations to the aforementioned paper. \par
The affine group $\mathrm{Aff}$ acts on its Lie algebra $\mathfrak{aff}$ through the \textit{adjoint action} 
\begin{equation}
\label{adjoint_action_equation}
    \mathrm{Ad}_{(x,a)}(X) := \begin{pmatrix} u & av - xu \\ 0 & 0 \end{pmatrix}, \qquad X = \begin{pmatrix} u & v \\ 0 & 0 \end{pmatrix} \in \mathfrak{aff}, \, (x,a) \in \mathrm{Aff}.
\end{equation}
A representation $\Phi$ of a Lie group $G$ on a vector space $V$ is always accompanied by a representation $\Phi^{*}$ of $G$ on the dual space $V^{*}$ defined by \[\langle \Phi(g)^{*}\eta, v \rangle := \langle \eta, \Phi(g^{-1})v \rangle, \quad g \in G, \, v \in V, \, \eta \in V^{*},\] where the bracket denotes the natural pairing between $V$ and $V^{*}$. In the case of the adjoint action in \eqref{adjoint_action_equation} we denote the accompanied representation on $\mathfrak{aff}^{*}$ by $\mathrm{Ad}^{*}$ and call it the \textit{coadjoint representation} of the affine group. We can realize $\mathfrak{aff}^{*}$ as matrices on the form \[\mathfrak{aff}^{*} \simeq \left\{(x,y) := \begin{pmatrix} x & 0 \\ y & 0\end{pmatrix}\Big|\, x,y \in \mathbb{R}\right\}.\] \par Any point on the form $(x,0) \in \mathfrak{aff}^{*}$ is a fixed point for the coadjoint representation. The upper and lower half-planes \[\mathcal{H}_{+} := \left\{(x,y) \in \mathfrak{aff}^{*}\Big|\, y > 0\right\}, \qquad \mathcal{H}_{-} := \left\{(x,y) \in \mathfrak{aff}^{*}\Big|\, y < 0\right\},\] both constitute distinct orbits. For reasons of symmetry it suffices to understand the representation corresponding to $\mathcal{H}_{+}$. It is convenient to identify $\mathcal{H}_{+} \simeq \mathrm{Aff}$ as sets and use the notation $(x,a)$ for a general element in $\mathcal{H}_{+}$. From general coadjoint orbit theory \cite[Chapter 1.2]{kirillov2004lectures} it follows that $\mathrm{Aff}$ is equipped with a canonical symplectic structure. In fact, this symplectic structure is simply the right Haar measure $a^{-1} \, da \, dx$ on $\mathrm{Aff}$. \par 
The main idea of Kirillov's theory is to associate irreducible representations of the Lie group to orbits of the coadjoint representation in a one-to-one manner. A realization of the representation corresponding to $\mathcal{H}_{+}$ is given by acting on $\psi \in L^{2}(\mathbb{R}_{+})$ by 
\begin{equation}
\label{Fourier_side_action}
    U(x,a)\psi(r) := e^{2\pi i x r}\psi(ar) = \frac{1}{\sqrt{a}}M_{x}D_{\frac{1}{a}}\psi(r).
\end{equation}
The representation $U$ is (up to a normalization) the representation \eqref{space_side_representation} on the Fourier side. Define the \textit{Stratonovich} -\textit{Weyl operator} on $L^{2}(\mathbb{R}_{+})$ by the formula 
\begin{equation*}
    \Omega(x,a)\psi(r) := a \int_{\mathbb{R}^2}e^{-2\pi i (xu + av)}U\left(\frac{ve^u}{\lambda(u)},e^{u}\right)\psi(r) \, du \, dv,
\end{equation*}
where $\psi \in L^{2}(\mathbb{R}_{+})$, $(x,a) \in \mathrm{Aff}$, and $\lambda$ is the function defined in \eqref{lambda_function}. The following result is given in \cite[Corollary 4.3]{gayral2007fourier}.

\begin{proposition}
\label{quantization_isometry_proposition}
There is an isometric isomorphism between $L_{r}^{2}(\mathrm{Aff})$ and the space of Hilbert-Schmidt operators on $L^{2}(\mathbb{R}_{+})$. The isomorphism sends $f \in L_{r}^{2}(\mathrm{Aff})$ to the operator $A_{f}$ on $L^{2}(\mathbb{R}_{+})$ defined by 
\begin{equation*}
A_{f}\psi(r) := \int_{-\infty}^{\infty} \int_{0}^{\infty}f(x,a)\Omega(x,a)\psi(r) \, \frac{da \, dx}{a}.
\end{equation*}
\end{proposition}

The association $f \mapsto A_f$ is called \textit{affine quantization}, while the direction $A_f \mapsto f$ is referred to as \textit{affine dequantization}. Moreover, we call $f$ the \textit{(affine) symbol} of $A_f$. Recall that any Hilbert-Schmidt operator $A$ on $L^{2}(\mathbb{R}_{+})$ has an associated  integral kernel $A_{K} \in L^{2}(\mathbb{R}_+ \times \mathbb{R}_+)$ so that \[A\psi(r) = \int_{0}^{\infty}A_{K}(r,s)\psi(s) \, \frac{ds}{s},\] for all $\psi \in L^{2}(\mathbb{R}_{+}).$ If $A = A_{f}$, then one can recover $f \in L_{r}^{2}(\mathrm{Aff})$ from the formula \[f(x,a) = \int_{-\infty}^{\infty}A_{K}\left(a\lambda(u),a\lambda(-u)\right)e^{-2\pi i x u} \, du.\]
Motivated by \eqref{Wigner_as_dequantization}, the affine Wigner distribution should be defined as the affine dequantization of a rank-one operator. Hence we have the following definition.

\begin{definition}
The \textit{affine cross-Wigner transform} acts on functions $\psi,\phi \in L^{2}(\mathbb{R}_{+})$ by 
\begin{align*}
W_{\mathrm{Aff}}^{\psi,\phi}(x,a) & := \int_{-\infty}^{\infty}\psi(a\lambda(u))\overline{\phi(a\lambda(-u))}e^{-2\pi i x u} \, du \\ & = \int_{-\infty}^{\infty}\psi\left(\frac{aue^u}{e^u - 1}\right)\overline{\phi\left(\frac{au}{e^{u} - 1}\right)}e^{-2\pi i x u} \, du,    
\end{align*} 
for $(x,a) \in \mathrm{Aff}$. We will refer to the diagonal $W_{\mathrm{Aff}}^{\psi} := W_{\mathrm{Aff}}^{\psi,\psi}$ as the \textit{affine Wigner distribution} of $\psi$.
\end{definition}

If $f \in L_{r}^{2}(\mathrm{Aff})$ is the symbol of the Hilbert-Schmidt operator $A_{f}$ acting on $L^{2}(\mathbb{R}_{+}),$ then 
\begin{equation}
\label{affine_weyl_correspondence}
    \left\langle A_{f}\psi,\phi \right\rangle_{L^{2}(\mathbb{R}_{+})} = \left\langle f,W_{\mathrm{Aff}}^{\phi,\psi}\right\rangle_{L_{r}^{2}(\mathrm{Aff})}, \quad \psi,\phi \in L^{2}(\mathbb{R}_{+}).
\end{equation}

\section{Basic Properties}
\label{sec: Baisc_Properties}

We begin by deriving basic properties of the affine Wigner distribution. Different approaches to the affine Wigner distribution will be discussed in Section \ref{sec: Alternative_Descriptions}. The affine Wigner distribution is fundamentally related to the transformation 
\begin{align*}
    \Pi:L^{2}(\mathbb{R}^+\times\mathbb{R}^+,(rs)^{-1} \, dr \, ds) & \to L_{r}^{2}(\mathrm{Aff}) \\ 
    \Pi(F)(u,a) & := F(a\lambda(u),a\lambda(-u)),
\end{align*}
where $\lambda$ is the function given in \eqref{lambda_function}. 

\begin{lemma}
\label{factorization}
The transformation $\Pi:L^{2}(\mathbb{R}^+\times\mathbb{R}^+,(rs)^{-1} \, dr \, ds) \to L_{r}^{2}(\mathrm{Aff})$ is an isometry. Moreover, the affine Wigner distribution $W_{\mathrm{Aff}}$ can be factorized as \[W_{\mathrm{Aff}}^{\psi,\phi}(x,a) = \mathcal{F}_1 \Pi(\psi \otimes \overline{\phi}) :=  \mathcal{F}_1 \left(\Pi(\psi \otimes \overline{\phi})(\cdot,a)\right)(x), \quad \psi,\phi \in L^{2}(\mathbb{R}_+),\] where $\mathcal{F}_1$ denotes the Fourier transform in the first component and $\psi \otimes \overline{\phi}(r,s) := \psi(r)\overline{\phi(s)}$ for $r,s \in \mathbb{R}_+$.
\end{lemma}

The factorization in Lemma \ref{factorization} is key for understanding essential properties of the affine Wigner distribution. We illustrate its use by extending the orthogonality property of the classical Wigner distribution in \eqref{usual_Wigner_orthogonality} to the affine setting.

\begin{proposition}
\label{affine_orthogonality_proposition}
The affine Wigner distribution satisfies the orthogonality relation
\begin{equation}
\label{affine_orthogonality_relation_new}
    \int_{-\infty}^{\infty}\int_{0}^{\infty}W_{\mathrm{Aff}}^{\psi}(x,a)\overline{W_{\mathrm{Aff}}^{\phi}(x,a)} \, \frac{da \, dx}{a} = |\langle \psi,\phi \rangle|^{2},
\end{equation}
for $\psi,\phi \in L^{2}(\mathbb{R}_+)$.
\end{proposition}

\begin{proof}
We use the factorization in Lemma \ref{factorization} and obtain 
\begin{align*}
    \left\langle W_{\mathrm{Aff}}^{\psi}, W_{\mathrm{Aff}}^{\phi} \right\rangle_{L_{r}^{2}(\mathrm{Aff})} & = \left\langle \mathcal{F}_1 \Pi\left(\psi \otimes \overline{\psi}\right), \mathcal{F}_1 \Pi\left(\phi \otimes \overline{\phi}\right) \right\rangle_{L_{r}^{2}(\mathrm{Aff})} \\ & = \left\langle \Pi\left(\psi \otimes \overline{\psi}\right), \Pi\left(\phi \otimes \overline{\phi}\right) \right\rangle_{L_{r}^{2}(\mathrm{Aff})} \\ & = \left\langle \psi \otimes \overline{\psi}, \phi \otimes \overline{\phi} \right\rangle_{L^{2}(\mathbb{R}^+\times\mathbb{R}^+,(rs)^{-1} \, dr \, ds)} \\ & = |\langle \psi,\phi \rangle|^{2}.
\end{align*}
\end{proof}

We will refer to \eqref{affine_orthogonality_relation_new} as the \textit{affine orthogonality relation} motivated by the analogous result for the classical Wigner distribution in \eqref{usual_Wigner_orthogonality}. Through a different (but ultimately equivalent) approach to the affine Wigner distribution taken in \cite{bertrand1992class} and \cite{molnar2001coherent}, the affine orthogonality relation is already known. The usefulness of the affine orthogonality relation can be readily demonstrated.

\begin{corollary}
We have equality $W_{\mathrm{Aff}}^{\psi} = W_{\mathrm{Aff}}^{\phi}$ for $\psi,\phi \in L^{2}(\mathbb{R}_+)$ if and only if $\psi = c \cdot \phi$ with $|c| = 1$.
\end{corollary}

\begin{proof}
It is clear from the definition of $W_{\mathrm{Aff}}$ that $\psi = c \cdot \phi$ with $|c| = 1$ implies that $W_{\mathrm{Aff}}^{\psi} = W_{\mathrm{Aff}}^{\phi}$. Conversely, if we assume that $W_{\mathrm{Aff}}^{\psi} = W_{\mathrm{Aff}}^{\phi}$ then the affine orthogonality relation \eqref{affine_orthogonality_relation_new} shows that \[|\langle \psi, \phi \rangle|_{L^{2}(\mathbb{R}_+)}^2 = \|\psi\|_{L^{2}(\mathbb{R}_+)}^4 = \|\phi\|_{L^{2}(\mathbb{R}_+)}^4.\] Hence $\|\phi\|_{L^{2}(\mathbb{R}_+)} = \|\psi\|_{L^{2}(\mathbb{R}_+)}$ and $|\langle \psi, \phi \rangle|_{L^{2}(\mathbb{R}_+)} = \|\psi\|_{L^{2}(\mathbb{R}_+)} \|\phi\|_{L^{2}(\mathbb{R}_+)}$. This can only happen when $\psi = c \cdot \phi$ with $|c| = 1$.
\end{proof}

The \textit{marginal properties} \cite[Lemma 4.3.6]{grochenig2013foundations} for the classical Wigner distribution strengthen a quantum mechanical interpretation of the Wigner distribution. For the affine Wigner distribution, we need an analogue of the Fourier transform on the group $\mathbb{R}^{+}$. This is the \textit{Mellin transform} given by \[\mathcal{M}(\psi)(x) = \mathcal{M}_{a}(\psi)(x) := \int_0^\infty \psi(a)a^{-2\pi i x} \, \frac{da}{a},\] for $x \in \mathbb{R}$ and $\psi \in L^{2}(\mathbb{R}_+)$. There is little consensus regarding the exponent of $a$ in the literature and we recommend checking carefully which convention is used whenever the Mellin transform is encountered. The Mellin transform is related to the Fourier transform $\mathcal{F}$ by $\mathcal{M}(\psi)(x) = \mathcal{F}(\Psi)(x)$, where $\Psi(x) := \psi(e^x)$ for $\psi \in L^{2}(\mathbb{R}_+)$. Hence the Mellin transform is a unitary map between the Hilbert spaces $L^{2}(\mathbb{R}_+)$ and $L^{2}(\mathbb{R})$. Additionally, the inverse of the Mellin transform is given by 
\begin{equation}
\label{Mellin_inverse}
    \mathcal{M}^{-1}(f)(a) = \mathcal{M}_{x}^{-1}(f)(a) = \int_{-\infty}^{\infty}f(x)a^{2\pi i x} \, dx,
\end{equation}
for $a \in \mathbb{R}_+$ and  $f \in L^{2}(\mathbb{R})$. Finally, the Mellin transform of a dilated function can be calculated to be 
\begin{equation}
\label{Mellin_dilation}
    \mathcal{M}(D_r \psi)(x) = r^{-2 \pi i x - \frac{1}{2}}\mathcal{M}(\psi)(x).
\end{equation}

The following marginal properties have been stated in \cite{shenoy1995wide} where the proofs are referred to the unpublished Ph.D. thesis of R. G. Shenoy. We provide a new proof of this remarkable fact to fill inn gaps that are lacking in the original sources.

\begin{proposition}
\label{marginal_properties_proposition}
The affine Wigner distribution satisfies the marginal properties
\begin{equation*}
    \int_{-\infty}^{\infty} W_{\mathrm{Aff}}^{\psi}(x,a) \, dx = |\psi(a)|^2, \quad (x,a) \in \mathrm{Aff},
\end{equation*} 
\begin{equation*}
    \int_{0}^{\infty} W_{\mathrm{Aff}}^{\psi}(x,a) \, \frac{da}{a} = |\mathcal{M}(\psi)(x)|^2, \quad (x,a) \in \mathrm{Aff},
\end{equation*} 
for all $\psi \in \mathcal{S}(\mathbb{R}_+)$.
\end{proposition}

\begin{proof}
To show the first marginal property, we again use the factorization in Lemma \ref{factorization} and obtain
\begin{align*}
    \int_{-\infty}^{\infty} W_{\mathrm{Aff}}^{\psi}(x,a) \, dx & = \mathcal{F}_{1}^{-1}\left(W_{\mathrm{Aff}}^{\psi}\right)(0,a)  \\ & = \mathcal{F}_{1}^{-1}\left(\mathcal{F}_1 \Pi\left(\psi \otimes \overline{\psi}\right)\right)(0,a) \\ & = \Pi\left(\psi \otimes \overline{\psi}\right)(0,a) \\ & = |\psi(a)|^2.
\end{align*}
In the last step, we used that $\lambda(0) = 1$ as a straightforward limit argument shows. The validity of the pointwise convergence in the Fourier inversion step is clear from the assumption on $\psi$. \par
For the second marginal property, we utilize a change of variables in the definition of the affine Wigner distribution to get the alternative form \begin{equation*}
W_{\mathrm{Aff}}^{\psi,\phi}(x,a) = \int_0^{\infty}u^{-2\pi ix} \psi\left(a\frac{u\log(u)}{u-1}\right)\overline{\phi\left(a\frac{\log(u)}{u-1}\right)} \, \frac{du}{u}.
\end{equation*}
The isometry property of the Mellin transform can then be used to obtain 
\begin{align*}
\int_0^\infty W_{\mathrm{Aff}}^\psi(x,a) \, \frac{da}{a}&=\int_0^\infty\int_0^\infty u^{-2\pi ix}\psi\left(a\frac{u\log(u)}{u-1}\right)\overline{\psi\left(a\frac{\log(u)}{u-1}\right)} \, \frac{da \, du}{au}\\
&=\int_0^\infty\int_{-\infty}^\infty u^{-2\pi ix}\mathcal{M}_a\left(\psi\left(a\frac{u\log(u)}{u-1}\right)\right)(\beta)\overline{\mathcal{M}_a\left(\psi\left(a\frac{\log(u)}{u-1}\right)\right)(\beta)} \, \frac{d\beta \, du}{u}.
\end{align*}
By using the dilation relation \eqref{Mellin_dilation} we can transform this expression to 
\begin{align*}
\int_0^\infty W_{\mathrm{Aff}}^\psi(x,a) \, \frac{da}{a}&=\int_0^\infty\int_{-\infty}^\infty u^{-2\pi ix}\left(\frac{u\log(u)}{u-1}\right)^{2\pi i\beta}\mathcal{M}_a\left(\psi\right)(\beta)\overline{\left(\frac{\log(u)}{u-1}\right)^{2\pi i\beta}\mathcal{M}_a(\psi)(\beta)} \, \frac{d\beta \, du}{u}\\
&=\int_0^\infty\int_{-\infty}^\infty u^{-2\pi ix}u^{2\pi i\beta}|\mathcal{M}_a(\psi)(\beta)|^2 \, \frac{d\beta \, du}{u}.
\end{align*}
Finally, by using the inverse Mellin transform \eqref{Mellin_inverse} we end up with 
\begin{align*}
    \int_0^\infty W_{\mathrm{Aff}}^\psi(x,a) \, \frac{da}{a}&=\int_0^\infty u^{-2\pi ix}\mathcal{M}^{-1}_\beta(|\mathcal{M}_a(\psi)(\beta)|^2)(u) \, \frac{du}{u}\\
&=\mathcal{M}_u(\mathcal{M}^{-1}_\beta(|\mathcal{M}_a(\psi)(\beta)|^2)(u))(x) \\ & = |\mathcal{M}(\psi)(x)|^2.
\end{align*}
Interchanging the order of integration and the pointwise convergence of the Mellin transform is easily justified under the assumption that $\psi \in \mathcal{S}(\mathbb{R}_+)$.
\end{proof}

\begin{remark}
It follows from Proposition \ref{marginal_properties_proposition} that \[\int_{0}^{\infty} \int_{-\infty}^{\infty} W_{\mathrm{Aff}}^{\psi}(x,a) \, \frac{da \, dx}{a} = \int_{0}^{\infty}|\psi(a)|^2 \, \frac{da}{a} = \|\psi\|_{L^{2}(\mathbb{R}_{+})}^2,\]
for all $\psi$ in the dense subspace $\mathcal{S}(\mathbb{R}_+) \subset L^{2}(\mathbb{R}_+)$. If $\|\psi\|_{L^{2}(\mathbb{R}_+)} = 1$ and $W_{\mathrm{Aff}}^\psi$ is everywhere non-negative, then the affine Wigner distribution would be a probability density function on the upper half-plane. We will elaborate on this in Section \ref{sec: Open_Questions}. 
\end{remark}

If $\psi \in \mathcal{S}(\mathbb{R}_+)$ have compact support and $a \in \mathbb{R}_+$ is outside the support of $\psi$, then Proposition \ref{marginal_properties_proposition} shows that\[\int_{-\infty}^{\infty} W_{\mathrm{Aff}}^{\psi}(x,a) \, dx = 0.\] This extreme case can be improved with the following \textit{finite support property}.

\begin{proposition}
Assume a continuous function $\psi \in L^{2}(\mathbb{R}_{+})$ is supported in $[r,s] \subset \mathbb{R}_{+}$. Then $W_{\mathrm{Aff}}^{\psi}(x,a) = 0$ for all $x \in \mathbb{R}$ whenever $a \notin [r,s].$
\end{proposition}
\begin{proof}
The functions $\psi(a\lambda(u))$ and $\psi(a\lambda(-u))$ are both non-zero if and only if \[\lambda(u),\lambda(-u) \in L := \left[\frac{r}{a},\frac{s}{a}\right].\] If $a > s$ then $L \subset (0,1)$. Hence it suffices to show that $\lambda(u)$ and $\lambda(-u)$ can not take values in $(0,1)$ simultaneously. This follows since $\lambda(u)$ is an increasing function that only takes values in $(0,1)$ whenever $u < 0$. If $a < r$ then $L \subset (1,\infty)$. In this case, the result follows from the fact that $\lambda(u) > 1$ if and only if $u > 0$. 
\end{proof}

\section{Alternative Descriptions}
\label{sec: Alternative_Descriptions}

Although the affine Wigner distribution was constructed rather recently, it has appeared in the literature several times in different disguises. In this section, we outline two instances of this and see how this enriches our understanding of the more subtle properties of the affine Wigner distribution. \par
Consider a function $\psi \in L^{2}(\mathbb{R}) \cap L^{2}(\mathbb{R}_+)$ that is supported on $\mathbb{R}_+$ and let $f \in L^{2}(\mathbb{R})$ be such that $\hat{f} = \psi$. The affine Wigner distribution $W_{\mathrm{Aff}}^{\phi}$ is related to the Bertrand $P := (P_{0},1)$ distribution described in \cite{papandreou1998quadratic} by the formula 
\begin{equation*}
    W_{\mathrm{Aff}}^{\psi}(x,a) = \frac{1}{a}Pf\left(a,-\frac{x}{a}\right).
\end{equation*}
Notice that the convention for the Bertrand distribution is that the first entry is a positive number, while the second is real. The distribution $P$ is called the \textit{Bertrand} $P_{0}$ \textit{distribution} and is in both the \textit{affine class} and the \textit{hyperbolic class} described in \cite{papandreou1998quadratic}. From this we can gauge several invariance properties of the affine Wigner distribution: 

\begin{itemize}
    \item The fact that $P$ is in the affine class gives the invariance properties 
    \begin{equation}
    \label{elementary_symmetries_of_affine_Wigner}
        W_{\mathrm{Aff}}^{M_{\omega}\psi}(x,a) = W_{\mathrm{Aff}}^{\psi}(x + a\omega,a), \qquad W_{\mathrm{Aff}}^{D_{r}\psi}(x,a) = \frac{1}{r}W_{\mathrm{Aff}}^{\psi}\left(x,\frac{a}{r}\right),
    \end{equation}
    where $M_{\omega}$ denotes the frequency-shift operator and $D_r$ denotes the dilation operator. These invariance properties can be summarized as 
    \begin{equation}
    \label{invariance_property}
        W_{\mathrm{Aff}}^{U(x,a)\psi}(y,b) = W_{\mathrm{Aff}}^{\psi}(y-bx, ab),
    \end{equation}
    where $U$ is the action of the affine group on $L^{2}(\mathbb{R}_+)$ given in \eqref{Fourier_side_action}.
    \item The fact that $P$ is in the hyperbolic class gives the invariance property 
    \begin{equation}
    \label{hyperbolic_invariance}
        W_{\mathrm{Aff}}^{\mathcal{H}(c,f_{r})\psi}(x,a) = W_{\mathrm{Aff}}^{\psi}(x+c,a),
    \end{equation}
    where $\mathcal{H}(c,f_{r})$ is the transformation \[\mathcal{H}(c,f_{r})\psi(r) := e^{-2\pi i c \ln\left(\frac{r}{f_{r}}\right)}\psi(r), \quad r,f_r > 0, \,  c \in \mathbb{R}.\]
    Pay attention to the fact that the \textit{positive reference frequency} $f_r$ only appears on the left-hand side of \eqref{hyperbolic_invariance}.
\end{itemize}
\par The affine Wigner distribution $W_{\mathrm{Aff}}$ can be derived in another way by emphasizing invariance properties as done in \cite{bertrand1992class} and \cite{molnar2001coherent}. From this perspective, one starts with a general quadratic distribution and require invariance under a group extension of the affine group. This will produce the distribution  \[W^{\psi}(x,a) := \int_{-\infty}^{\infty}\psi(a\lambda(u))\overline{\psi(a\lambda(-u))}e^{-2\pi iux}\mu(u) \, du,\] where $\mu(u)$ is a weight function that satisfies $\overline{\mu(u)} = \mu(-u)$. The requirement that $W^{\psi}$ satisfies the affine orthogonality relation
\begin{equation*}
\int_{-\infty}^{\infty}\int_{0}^{\infty}W^{\psi}(x,a)\overline{W^{\phi}(x,a)} \, \frac{da \, dx}{a} = |\langle \psi,\phi \rangle|^{2},
\end{equation*} 
forces $\mu \equiv 1$ so that $W^{\psi} = W_{\mathrm{Aff}}^{\psi}$. Through this description, the affine orthogonality relation would not need to be proved as it is incorporated in the construction. However, this approach conceals the quantization picture. The affine Wigner distribution $W_{\mathrm{Aff}}$ is a special case of a family of distributions that are called \textit{tomographic distributions} in \cite{bertrand1992class}. \par
We would like to mention that there have been other attempts at defining a notion of affine Wigner distribution that do not coincide with our definition. As an example, we refer the reader to \cite{gazeau2016covariant} and the recent successor paper \cite{gazeau20192d} where an affine Wigner-like quasi-probability is defined through a semi-classical quantization approach. Although this is different from the approach in \cite{gayral2007fourier} that our work is based on, it has similarities in both motivation and properties. 

\section{Affine Convolution Representation of the Scalogram}
\label{sec: Convolution_Relation}

Recall from the introduction that the classical Wigner distribution can represent the spectogram through convolution
\begin{equation}
\label{scalogram_convolution_second_time}
    \mathrm{SPEC}_{g}f(x,\omega) = W_{P(g)} * W_{f}(x,\omega) = W_{P(\hat{g})} * W_{\hat{f}}(\omega,-x),
\end{equation}
where $P(g)(x) := g(-x)$. This relation was mentioned in \cite[Eq 85]{mark1970spectral} where the Wigner distribution went under the name \textit{(instantaneous) spectrum-smoothing function}. It later appeared in \cite[Eq 4.5]{classen1970}, where it was used to show that the spectogram is a Cohen class distribution. Finally, it was put on more rigorous foundations in \cite[Proposition 1.99]{folland1989harmonic}. \par 
One could also use the classical Wigner distribution to represent the scalogram by writing 
\begin{align*}
    \mathrm{SCAL}_{g}{f}(x,a) & = \langle f,T_x D_a g \rangle_{L^{2}(\mathbb{R})} \overline{\langle f,T_x D_a g \rangle}_{L^{2}(\mathbb{R})} \\ & = \langle W_{f}, W_{T_x D_a g} \rangle_{L^{2}(\mathbb{R}^{2})} \\ & = \int_{\mathbb{R}^2}W_{f}(\tau, \xi)W_{g}\left(\frac{\tau -x}{a},a\xi\right) \, d\tau \, d\xi,
\end{align*}
where we used the orthogonality of the classical Wigner distribution given in \eqref{usual_Wigner_orthogonality}. However, even though this might look like a convolution type representation of the scalogram, it does not incorporate one of the natural measures on the affine group. Hence we need to look elsewhere to obtain a proper convolution representation for the scalogram. \par 
Before stating the result, we recall some generalities from the theory of locally compact groups applied to the affine group: The \textit{affine convolution} between two functions $f,g$ on the affine group is given (whenever it is well-defined) by 
\begin{align*}
    f *_{\mathrm{Aff}} g(x,a) & := \int_{-\infty}^{\infty}\int_{0}^{\infty}f(y,b)g\left((y,b)^{-1} \cdot_{\mathrm{Aff}} (x,a)\right) \, \frac{db \, dy}{b^2} \\ & = \int_{-\infty}^{\infty}\int_{0}^{\infty}g(y,b)f\left((x,a) \cdot_{\mathrm{Aff}} (y,b)^{-1}\right) \, \frac{db \, dy}{b}.
\end{align*}
A departure from the usual Euclidean convolution is that the affine convolution is not commutative. The \textit{modular function} $\Delta$ on any locally compact group measures the difference between the right and left Haar measure. We refer the reader to a precise definition in \cite[Chapter 2.4]{folland2016course} as we only need that the modular function on the affine group is \[\Delta(x,a) = \frac{1}{a}, \qquad (x,a) \in \mathrm{Aff}.\] Finally, the \textit{(right) involution} of a function $f$ on the affine group is given by \[I(f)(x,a) := \Delta(x,a) \overline{f((x,a)^{-1})} = \frac{1}{a}\overline{f\left(-\frac{x}{a},\frac{1}{a}\right)}, \qquad (x,a) \in \mathrm{Aff}.\]
The following convolution formula should be compared with \eqref{scalogram_convolution_second_time}.

\begin{theorem}
\label{convolution_relation_theorem}
Let $f$ and $g$ be square integrable functions on the real line and assume their Fourier transforms $\phi := \widehat{f}$ and $\psi := \widehat{g}$ are supported in $\mathbb{R}_{+}$ and are in $L^{2}(\mathbb{R}_{+})$. Then\[\mathrm{SCAL}_{g}{f}(x,a) = \left(I\left(W_{\mathrm{Aff}}^{\psi}\right) *_{\mathrm{Aff}} \Delta W_{\mathrm{Aff}}^{\phi}\right)\left(\frac{x}{a}, \frac{1}{a}\right),\]
for all $(a,x) \in \mathrm{Aff}$.
\end{theorem}

\begin{proof}
By using Parseval's identity and that the support of the Fourier transforms are in $\mathbb{R}_{+}$ we obtain 
\begin{align*}
    \mathrm{SCAL}_{g}{f}(x,a) & = \mathcal{W}_{g}f(x,a)\overline{\mathcal{W}_{g}f(x,a)} \\ & = \left|\langle f,T_x D_a g\rangle_{L^{2}(\mathbb{R})}\right|^2 \\ & = \left|\langle \phi,\sqrt{a}U(x,a)\psi\rangle_{L^{2}(\mathbb{R}_+)}\right|^2,
\end{align*}
where $U(x,a)$ is the action of the affine group given in \eqref{Fourier_side_action}.
The affine orthogonality relation given in Proposition \ref{affine_orthogonality_proposition} and the invariance property given in \eqref{invariance_property} together show that 
\begin{align*}
    \mathrm{SCAL}_{g}{f}(x,a) & = \int_{-\infty}^{\infty}\int_{0}^{\infty} W_{\mathrm{Aff}}^{\phi}(y,b) \cdot a \cdot W_{\mathrm{Aff}}^{U(x,a)\psi}(y,b) \, \frac{db \, dy}{b} \\ & =  \int_{-\infty}^{\infty}\int_{0}^{\infty}W_{\mathrm{Aff}}^{\phi}(y,b)\cdot a \cdot W_{\mathrm{Aff}}^{\psi}(y-bx,ab) \, \frac{db \, dy}{b} \\ & =
    \int_{-\infty}^{\infty}\int_{0}^{\infty}W_{\mathrm{Aff}}^{\phi}(y,b) \cdot ab \cdot W_{\mathrm{Aff}}^{\psi}\left((y,b)\cdot_{\mathrm{Aff}}(-x,a)\right) \, \frac{db \, dy}{b^2}.
\end{align*}
We use the involution on the affine group to write \[ab \cdot  W_{\mathrm{Aff}}^{\psi}\left((y,b)\cdot_{\mathrm{Aff}}(-x,a)\right) = I\left(W_{\mathrm{Aff}}^{\psi}\right)\left((-x,a)^{-1} \cdot_{\mathrm{Aff}} (y,b)^{-1}\right).\] Combining these observations shows that 
\begin{align*}
    \mathrm{SCAL}_{g}{f}(x,a) & = \int_{-\infty}^{\infty}\int_{0}^{\infty}\frac{W_{\mathrm{Aff}}^{\phi}(y,b)}{b} \cdot I\left(W_{\mathrm{Aff}}^{\psi}\right)\left(\left(\frac{x}{a},\frac{1}{a}\right) \cdot_{\mathrm{Aff}} (y,b)^{-1}\right) \, \frac{db \, dy}{b} \\ & = \left(I\left(W_{\mathrm{Aff}}^{\psi}\right) *_{\mathrm{Aff}} \Delta W_{\mathrm{Aff}}^{\phi}\right)\left(\frac{x}{a},\frac{1}{a}\right),
\end{align*} where $\Delta$ is the modular function on the affine group.
\end{proof}

\section{The Affine Ambiguity Function and Distributional Extension of Affine Quantization}
\label{sec: Affine_ambiguity_function_section}

The \textit{cross-ambiguity function} in time-frequency analysis of $f,g \in L^{2}(\mathbb{R})$ is defined to be 
\begin{equation*}
    A(f,g)(x,\omega) := \int_{-\infty}^{\infty}f\left(t + \frac{x}{2}\right)\overline{g\left(t - \frac{x}{2}\right)}e^{-2\pi i t \omega} \, dt, \qquad (x,\omega) \in \mathbb{R}^2.
\end{equation*}
The \textit{ambiguity function} $Af := A(f,f)$ of $f \in L^{2}(\mathbb{R})$ has been frequently used in radar applications \cite[Chapter 4.2]{grochenig2013foundations}. In the affine setting, we suggest that the \textit{affine cross-ambiguity function} of $\psi,\phi \in L^{2}(\mathbb{R}_+)$ should be the function $A_{\mathrm{Aff}}^{\psi,\phi}$ on the upper half-plane defined by 
\begin{equation*}
    A_{\mathrm{Aff}}^{\psi,\phi}(x,a) := \int_{0}^{\infty} \psi\left(r \sqrt{a}\right)\overline{\phi\left(\frac{r}{\sqrt{a}}\right)}r^{-2\pi i x} \, \frac{dr}{r}, \quad (x,a) \in \mathrm{Aff}.
\end{equation*}
Similarly as before, we call the function $A_{\mathrm{Aff}}^{\psi} := A_{\mathrm{Aff}}^{\psi, \psi}$ the \textit{affine ambiguity function}. \par 
In \cite{shenoy1995wide} the authors define a different notion of affine ambiguity function under the name \textit{wide-band ambiguity function}. Notice that the definition of $A_{\mathrm{Aff}}^{\psi,\phi}$ incorporates the Haar measure on $\mathbb{R}_{+}$ in a natural way. Moreover, we will show that our definition possesses properties that justifies the terminology \textit{affine ambiguity function}. The first statement in the following lemma is a straightforward change of variables, while the last statement is a direct consequence of \cite[Lemma 4.2.1]{grochenig2013foundations}.

\begin{lemma}
\label{lemma_affine_ambiguity_function}
   For $\psi,\phi \in L^{2}(\mathbb{R}_+)$ we define the functions $\Psi(x) := \psi(e^{x})$ and $\Phi(x) := \phi(e^{x})$ for $x \in \mathbb{R}$. Then $\Psi, \Phi \in L^{2}(\mathbb{R})$ and \[A_{\mathrm{Aff}}^{\psi,\phi}(\omega,e^x) = A(\Psi,\Phi)(x,\omega), \qquad (x,\omega) \in \mathbb{R}^2.\] Moreover, the affine ambiguity function satisfies
   \[|A_{\mathrm{Aff}}^{\psi}(x,a)| <  A_{\mathrm{Aff}}^{\psi}(0,1) = \|\psi\|_{L^{2}(\mathbb{R}_+)}^{2},\] for every $(x,a) \neq (0,1).$
\end{lemma}

From Lemma \ref{lemma_affine_ambiguity_function} we obtain most of the expected results regarding the affine cross-ambiguity function. To illustrate this, we show two essential properties.

\begin{proposition}
The affine cross-ambiguity function satisfies the orthogonality relation
\begin{equation*}
    \left\langle A_{\mathrm{Aff}}^{\psi_1, \phi_1}, A_{\mathrm{Aff}}^{\psi_2, \phi_2} \right\rangle_{L_{r}^{2}(\mathrm{Aff})} = \langle \psi_1, \psi_2 \rangle_{L^{2}(\mathbb{R}_+)} \overline{\langle \phi_1, \phi_2 \rangle}_{L^{2}(\mathbb{R}_+)},
\end{equation*}
for $\psi_1,\psi_2,\phi_1,\phi_2 \in L^{2}(\mathbb{R}_+)$.
\end{proposition}

\begin{proof}
   Let $\Psi_i(x) := \psi_i(e^x)$ and $\Phi_i(x) := \phi_i(e^x)$ for $i = 1,2$ and $x \in \mathbb{R}$. We have by Lemma \ref{lemma_affine_ambiguity_function} that \begin{align*}
    \left\langle A_{\mathrm{Aff}}^{\psi_1, \phi_1}, A_{\mathrm{Aff}}^{\psi_2, \phi_2} \right\rangle_{L_{r}^{2}(\mathrm{Aff})} & = \int_{- \infty}^{\infty} \int_{0}^{\infty} A_{\mathrm{Aff}}^{\psi_1, \phi_1}(x,a) \overline{A_{\mathrm{Aff}}^{\psi_2, \phi_2}(x,a)} \, \frac{da \, dx}{a} \\ & = \int_{- \infty}^{\infty} \int_{0}^{\infty} A(\Psi_1, \Phi_1)(\ln(a),x) \overline{A(\Psi_2, \Phi_2)(\ln(a),x)} \, \frac{da \, dx}{a} \\ & = \int_{- \infty}^{\infty} \int_{- \infty}^{\infty} A(\Psi_1, \Phi_1)(u,x) \overline{A(\Psi_2, \Phi_2)(u,x)} \, du\,dx.
   \end{align*}
   In \cite[Lemma 4.3.4]{grochenig2013foundations} it is showed that the ambiguity function is related to the usual cross-Wigner transform by 
   \begin{equation}
   \label{ambiguity_wigner_correspondance}
       W(\Psi_i, \Phi_i) = \mathcal{F} \mathcal{U}A(\Psi_i, \Phi_i), \quad i = 1,2, 
   \end{equation}
   where $\mathcal{F}$ is the Fourier transform and $\mathcal{U}$ is the rotation $\mathcal{U}F(x,\omega) := F(\omega,-x)$ for a function $F$ on $\mathbb{R}^2$. In particular, the ambiguity function satisfies the same orthogonality properties as the cross-Wigner transform \eqref{usual_Wigner_orthogonality}. Hence we obtain that \[\left\langle A_{\mathrm{Aff}}^{\psi_1, \phi_1}, A_{\mathrm{Aff}}^{\psi_2, \phi_2} \right\rangle_{L_{r}^{2}(\mathrm{Aff})} = \langle \Psi_1, \Psi_2 \rangle_{L^{2}(\mathbb{R})} \overline{\langle \Phi_1, \Phi_2 \rangle}_{L^{2}(\mathbb{R})} = \langle \psi_1, \psi_2 \rangle_{L^{2}(\mathbb{R}_+)} \overline{\langle \phi_1, \phi_2 \rangle}_{L^{2}(\mathbb{R}_+)}.\]
   
   \vspace{-0.90cm}
   
\end{proof}
\vspace{0.2cm}

The second property is an uncertainty principle for the affine ambiguity function. Notice that if $U = U_1 \times U_2 \subset \mathrm{Aff}$ is a Borel set, then the right Haar measure $\mu_{r}(U)$ of $U$ is given by \[\mu_{r}(U) = |U_{1} \times \ln(U_2)|,\] where $| \cdot |$ is the usual Euclidean measure.

\begin{corollary}
Let $\psi \in L^{2}(\mathbb{R}_+)$ be normalized and let $U \subset \mathrm{Aff}$ be a Borel set. Assume that there is an $\epsilon > 0$ such that 
\begin{equation}
\label{uncertainty_assumption}
    \iint_{U}|A_{\mathrm{Aff}}^{\psi}(x,a)|^2 \, \frac{da \, dx}{a} \geq 1 - \epsilon.
\end{equation}
Then \[\mu_{r}(U) \geq (1 - \epsilon)^{\frac{p}{p-2}}\left(\frac{p}{2}\right)^{\frac{2}{p-2}},\] for all $p > 2$. In particular, setting $p = 4$ gives $\mu_{r}(U) \geq 2(1 - \epsilon)^2$ while letting $p$ tend to infinity shows that $\mu_{r}(U) \geq 1 - \epsilon$.
\end{corollary}
\begin{proof}
   Notice that the assumption \eqref{uncertainty_assumption} is by Lemma \ref{lemma_affine_ambiguity_function} equivalent to \[\int_{U_1}\int_{\ln(U_2)} |A\Psi(u,x)|^2 \, du \, dx \geq 1 - \epsilon,\]
   where $\Psi(x) := \psi(e^x)$. We can write $A\Psi(u,x) = e^{\pi i u x}V_{\Psi}\Psi(u,x)$, where $V$ is the STFT given in \eqref{STFT}. The assumption \[\int_{U_1}\int_{\ln(U_2)} |V_{\Psi}{\Psi}(u,x)|^2 \, du \, dx \geq 1 - \epsilon\] implies by Lieb's uncertainty principle \cite[Theorem 3.3.3]{grochenig2013foundations} that we have \[\mu_{r}(U) = |U_1 \times \ln(U_2)| \geq (1 - \epsilon)^{\frac{p}{p-2}}\left(\frac{p}{2}\right)^{\frac{2}{p-2}}, \qquad p > 2.\]
   \vspace{-1.25cm}
   
\end{proof}
\vspace{0.2cm}

We need to relate the affine ambiguity function to the affine Wigner distribution. Define the function \[\Theta(y,b) := \left(\frac{\sqrt{b}\log(b)}{b - 1}\right)^{2\pi i y},\] for $y \in \mathbb{R}$ and $b > 0$ with the convention that $\Theta(y,1) = 1$ for all $y \in \mathbb{R}$. If we write $b = e^{u}$ for $u = \log(b)$, then \[\frac{\sqrt{b}\log(b)}{b - 1} = \sqrt{\lambda(u)\lambda(-u)},\] where $\lambda$ is the function given in \eqref{lambda_function}. Hence we can think of $\Theta(y,b)$ as arising from a symmetrization of the function $\lambda$. We leave the verification of the following result to the reader as it is straightforward.

\begin{lemma}
\label{Mellin_factorization}
   For $\psi,\phi \in L^{2}(\mathbb{R}_+)$ we have the equality
   \begin{equation*}
       W_{\mathrm{Aff}}^{\psi,\phi}(x,a) = \mathcal{M}_{y}^{-1} \otimes \mathcal{M}_{b}\left[\Theta(y,b) \cdot A_{\mathrm{Aff}}^{\psi,\phi}(y,b)\right](x,a),
   \end{equation*}
   where $(x,a) \in \mathrm{Aff}$ and $\mathcal{M}$ is the Mellin transform.
\end{lemma}

We say that a smooth function $f$ on the affine group $\mathrm{Aff}$ is \textit{rapidly decaying} if \[(x,\omega) \longmapsto f(x,e^{\omega}) \in \mathcal{S}(\mathbb{R}^2).\]
The space of rapidly decaying smooth functions on $\mathrm{Aff}$ will be denoted by $\mathcal{S}(\mathrm{Aff})$. The following result illustrates how we can use the Mellin transform and the affine ambiguity function to deduce properties of the affine Wigner distribution.

\begin{proposition}
\label{Wigner_Schwartz_functions}
For $\psi,\phi \in \mathcal{S}(\mathbb{R}_+)$ the affine Wigner distribution satisfies $ W_{\mathrm{Aff}}^{\psi,\phi} \in \mathcal{S}(\mathrm{Aff})$.
\end{proposition}

\begin{proof}
If we let $\Psi(x) := \psi(e^x)$ and $\Phi(x) := \phi(e^{x})$, then by Lemma \ref{lemma_affine_ambiguity_function} and Lemma \ref{Mellin_factorization} we want to show that \[(x,\omega) \longmapsto \mathcal{M}_{y}^{-1} \otimes \mathcal{M}_{b}\left[\Theta(y,b) \cdot A^{\Psi,\Phi}(\log(b),y)\right](x,e^{\omega}) \in \mathcal{S}(\mathbb{R}^2).\]
It follows from \cite[Theorem 11.2.5]{grochenig2013foundations} that the usual ambiguity function sends Schwartz functions on $\mathbb{R}$ to Schwartz functions on $\mathbb{R}^2$. Hence $A(y,b) := A^{\Psi,\Phi}(\log(b),y) \in \mathcal{S}(\mathrm{Aff})$. Since $\Theta(y,b)$ is a smooth function with polynomially bounded derivatives, the same goes for the product $\Theta(y,b) \cdot A(y, b)$. Recall that the Mellin transform is related to the Fourier transform by the formula $\mathcal{M}(\psi)(x) = \mathcal{F}(\Psi)(x)$ for $x \in \mathbb{R}$. Thus the result follows from the fact that the Fourier transform preserves Schwartz functions. 
\end{proof}

The dual space of $\mathcal{S}(\mathrm{Aff})$ will be denoted by $\mathcal{S}'(\mathrm{Aff})$ and called the \textit{tempered distributions on the affine group}. The following is now a direct consequence of \eqref{affine_weyl_correspondence} and Proposition \ref{Wigner_Schwartz_functions}.

\begin{corollary}
\label{extend_affine_quantization_to_tempered}
The affine quantization $f \mapsto A_{f}$ extends to a well-defined map from $f \in \mathcal{S}'(\mathrm{Aff})$ to operators $A_{f}:\mathcal{S}(\mathbb{R}_+) \to \mathcal{S}'(\mathbb{R}_+)$.
\end{corollary}

\begin{example}
\label{point_mass_example}
Consider the point measure $\delta_{\mathrm{Aff}}(x,a) \in \mathcal{S}'(\mathrm{Aff})$ defined by
\[\left\langle \delta_{\mathrm{Aff}}(x,a), f \right\rangle = \overline{f(x,a)}, \] for $f \in \mathcal{S}(\mathrm{Aff})$ and $(x,a) \in \mathrm{Aff}$. We compute for $\psi, \phi \in \mathcal{S}(\mathbb{R}_+)$ that 
\[\left\langle A_{\delta_{\mathrm{Aff}}(x,a)} \psi, \phi \right\rangle = \left\langle \delta_{\mathrm{Aff}}(x,a), W_{\mathrm{Aff}}^{\phi,\psi} \right\rangle = \overline{W_{\mathrm{Aff}}^{\phi,\psi}(x,a)} = W_{\mathrm{Aff}}^{\psi,\phi}(x,a).\]
Hence the operator $A_{\delta_{\mathrm{Aff}}(x,a)}$ is weakly defined through the values of the affine Wigner distribution.
\end{example}

\section{An Almost Analytic Decomposition}
\label{sec: Polyanalytic_Decomposition}

Recall that analytic and anti-analytic functions $f$ are characterized by the equations \[\partial_{\bar{z}}f(z) = 0, \qquad \partial_{z}f(z) = 0,\] respectively. The fact that the affine Wigner distribution $W_{\mathrm{Aff}}^{\psi,\phi}$ is in the space $L_{r}^{2}(\mathrm{Aff})$ for $\psi, \phi \in L^{2}(\mathbb{R}_+)$ allows us to exclude (anti-)analytic functions from being in the image of the affine Wigner distribution. 

\begin{proposition}
\label{no_analytic_functions}
There are no analytic or anti-analytic functions in the space $L_{r}^{2}(\mathrm{Aff})$. In particular, functions on the form $f = W_{\mathrm{Aff}}^{\psi,\phi}$ for $\psi,\phi \in L^{2}(\mathbb{R}_+)$ can neither be analytic nor anti-analytic. 
\end{proposition}

\begin{proof}
The conclusion is easier to obtain by looking at the isomorphic spaces in the unit disc $\mathbb{D}$ by applying the standard linear fractional transformation. Under this transformation, the analytic functions in $L_{r}^{2}(\mathrm{Aff})$ are transformed to the analytic functions $f$ in the unit disc satisfying the integrability condition
\begin{equation}
\label{disc_bargman_space}
\int_{\mathbb{D}}\frac{|f(z)|^2}{1-|z|^2} \, dz < \infty.
\end{equation}
Any such analytic function will have to vanish as it approaches the boundary circle. Thus they are identically zero inside the unit disc as well by the unique continuation principle for analytic functions. The case of anti-analytic functions is similar.
\end{proof}

\begin{remark}
Proposition \ref{no_analytic_functions} shows a big difference between the affine Wigner distribution and both the classical Wigner distribution and the wavelet transform; the classical Wigner distribution can produce Gaussians, while one can obtain plenty of analytic functions from the wavelet transform as shown in \cite[Chapter 2.5]{daubechies1992ten}.
\end{remark}

From Proposition \ref{no_analytic_functions} a few natural questions emerge: What kind of analytic-like functions are in the space $L_{r}^{2}(\mathrm{Aff})$? Is it possible to decompose the space $L_{r}^{2}(\mathrm{Aff})$ into pieces consisting of \textquote{almost analytic} and \textquote{almost anti-analytic} functions? By looking at the equivalent integrability condition in the disk, it is clear that the function $f(z) = 1 - |z|^2$ satisfies \eqref{disc_bargman_space}. Although it is not analytic nor anti-analytic, it is almost both.

\begin{definition}
A function $f:U \to \mathbb{C}$ where $U \subset \mathbb{C}$ is an open subset is called \textit{poly-analytic of order} $n \in \mathbb{N}$ if \[\partial_{\overline{z}}^{n}f = 0.\] Similarly, a function $f:U \to \mathbb{C}$ will be called \textit{anti-poly-analytic of order} $n \in \mathbb{N}$ if \[\partial_{z}^{n}f = 0.\]
\end{definition}

Notice that the function $f(z) = 1 - |z|^2$ is both poly-analytic and anti-poly-analytic of order two. The poly-analytic functions of order one are simply analytic functions, while the anti-poly-analytic functions of order one are the anti-analytic functions. We refer to an (anti-)poly-analytic function of order $n$ as \textit{pure} if it is not (anti-)poly-analytic of order $n-1$ or lower. 
Poly-analytic functions and anti-poly-analytic functions do not inherit all the amazing properties that analytic functions are known for; the function $f(z) = 1 - |z|^2$ vanish on the whole unit circle without being identically zero. The failure of the strong unique continuation principle for poly-analytic and anti-poly-analytic functions is what makes it possible for them to exist in $L_{r}^{2}(\mathrm{Aff})$.
\par 
We will show, inspired by a method in \cite{vasilevski1999structure}, that we can decompose the space $L_{r}^{2}(\mathrm{Aff})$ into pieces consisting of poly-analytic and anti-poly-analytic functions.\,\,Before we can do this, we explore how generalized Laguerre polynomials give us a suitable orthonormal basis.

\begin{definition}
The \textit{generalized Laguerre polynomials} with (real) \textit{parameter} $\alpha$ are the polynomials $L_{n}^{(\alpha)}$ defined by \[L_{n}^{(\alpha)}(x) := \frac{x^{-\alpha}e^{x}}{n!}\frac{d^n}{dx^n}\left(e^{-x}x^{n + \alpha}\right) = \sum_{k = 0}^{n}(-1)^{k}\binom{n + \alpha}{n-k}\frac{x^k}{k!}.\]
\end{definition}
For $\alpha > -1$ we have the orthogonality relation
\begin{equation}
\label{orthogonality_Laguerre}
    \int_{0}^{\infty}x^{\alpha}e^{-x}L_{n}^{(\alpha)}(x)L_{m}^{(\alpha)}(x) \, dx = \frac{\Gamma(n + \alpha + 1)}{n!}\delta_{n,m},
\end{equation} where $\Gamma$ denotes the Gamma function. Introduce the functions 
\begin{equation}
\label{explicit_Laguerre}
    \mathcal{L}_{n}^{(\alpha)}(x) := \sqrt{\frac{n!}{\Gamma(n+\alpha+1)}}x^{\frac{\alpha + 1}{2}}e^{-\frac{x}{2}}L_{n}^{(\alpha)}(x),
\end{equation}
for $\alpha > -1$. It is straightforward to check that the functions in \eqref{explicit_Laguerre} form an orthonormal basis for $L^{2}(\mathbb{R}_{+})$ for each fixed $\alpha > -1$ by using \eqref{orthogonality_Laguerre}. If $\alpha = 1$ we use the simplified notation $\mathcal{L}_n := \mathcal{L}_{n}^{(1)}$.

\begin{lemma}
\label{Wigner_distribution_bases}
If $\{\psi_{n}\}_{n \in \mathbb{N}}$ is an orthonormal basis for $L^{2}(\mathbb{R}_{+}),$ then the functions $\{W_{\mathrm{Aff}}^{\psi_{n},\psi_{m}}\}_{n,m \in \mathbb{N}}$ constitute an orthonormal basis for $L_{r}^{2}(\mathrm{Aff})$. In particular, for a fixed $\alpha > -1$, we can expand any $f \in L_{r}^{2}(\mathrm{Aff})$ as \[f = \sum_{n,m = 0}^{\infty}\left\langle f,W_{\mathrm{Aff}}^{\mathcal{L}_{n}^{(\alpha)},\mathcal{L}_{m}^{(\alpha)}}\right\rangle W_{\mathrm{Aff}}^{\mathcal{L}_{n}^{(\alpha)},\mathcal{L}_{m}^{(\alpha)}}.\]
\end{lemma}
\begin{proof}
The orthonormality of the functions $W_{\mathrm{Aff}}^{\psi_{n},\psi_{m}}$ clearly follows from Proposition \ref{affine_orthogonality_proposition}. To see the completeness in $L_{r}^{2}(\mathrm{Aff})$ we assume that $f \in L_{r}^{2}(\mathrm{Aff})$ satisfies \[\left\langle f, W_{\mathrm{Aff}}^{\psi_{n},\psi_{m}}\right\rangle_{L_{r}^{2}(\mathrm{Aff})} = 0\] for every $n,m \in \mathbb{N}$. If we let $A_{f}$ be the Hilbert-Schmidt operator acting on $L^{2}(\mathbb{R}_{+})$ corresponding to $f$ through the quantization procedure, then equation (\ref{affine_weyl_correspondence}) implies that \[\langle A_{f}\psi_{m},\psi_{n} \rangle_{L^{2}(\mathbb{R}_{+})} = 0.\] Since $\{\psi_{n}\}_{n \in \mathbb{N}}$ is an orthonormal basis for $L^{2}(\mathbb{R}_{+})$ we have that $A_{f} = 0.$ As the quantization correspondence between $f$ and $A_{f}$ is a Hilbert space isomorphism, we conclude that $f = 0$.
\end{proof}

Returning to the problem of decomposing $L_{r}^{2}(\mathrm{Aff})$, we use the notation $\mathbb{A}^{n}(\mathrm{Aff})$ and $\mathbb{A}^{\perp,n}(\mathrm{Aff})$ for all functions $f \in L_{r}^{2}(\mathrm{Aff})$ that are poly-analytic and anti-poly-analytic of order $n$, respectively. Finally, we use the notation $\mathcal{A}^{n}(\mathrm{Aff}) \subset \mathbb{A}^{n}(\mathrm{Aff})$ and $\mathcal{A}^{\perp,n}(\mathrm{Aff}) \subset \mathbb{A}^{\perp,n}(\mathrm{Aff})$ for the subspaces of pure poly-analytic and pure anti-poly-analytic functions of order $n$, respectively.

\begin{proposition}
\label{polyanalytic_decomposition}
The space $L_{r}^{2}(\mathrm{Aff})$ has the orthogonal decomposition \[L_{r}^{2}(\mathrm{Aff}) = \bigoplus_{n = 2}^{\infty}\mathcal{A}^{n}(\mathrm{Aff}) \oplus \mathcal{A}^{\perp,n}(\mathrm{Aff}).\] Moreover, the spaces $\mathbb{A}^{n}(\mathrm{Aff})$, $\mathbb{A}^{\perp,n}(\mathrm{Aff})$, $\mathcal{A}^{n}(\mathrm{Aff})$, and $\mathcal{A}^{\perp,n}(\mathrm{Aff})$ for $n \geq 2$ can be identified with the spaces \[\mathbb{A}^{n}(\mathrm{Aff}) \simeq L^{2}(\mathbb{R}_{+},dx) \otimes \bigoplus_{k = 0}^{n-2}\mathrm{span}\left\{\mathcal{L}_{k}\right\}, \qquad \mathbb{A}^{\perp,n}(\mathrm{Aff}) \simeq L^{2}(\mathbb{R}_{-},dx) \otimes \bigoplus_{k = 0}^{n-2}\mathrm{span}\left\{\mathcal{L}_{k}\right\},\] \[\mathcal{A}^{n}(\mathrm{Aff}) \simeq L^{2}(\mathbb{R}_{+},dx) \otimes \mathrm{span}\left\{\mathcal{L}_{n-2}\right\}, \qquad \mathcal{A}^{\perp,n}(\mathrm{Aff}) \simeq L^{2}(\mathbb{R}_{-},dx) \otimes \mathrm{span}\left\{\mathcal{L}_{n-2}\right\}.\]
\end{proposition}

We have delegated the proof of Proposition \ref{polyanalytic_decomposition} to Appendix \ref{Appendix_Proof} as it is heavily inspired by a technique used in \cite{vasilevski1999structure}. The poly-analytic functions have appeared prominently in the work of Abreu, see e.g. \cite{abreu2012super}, in the context of wavelet analysis and sampling theory. However, a significant difference is that Abreu only considers poly-analytic functions and not the anti-poly-analytic functions. 

\section{Applications}
\label{sec: Applications}

\subsection{Affine Wigner Approximation}

Let us use the notation \[\mathfrak{W}(\mathrm{Aff}) := \left\{W_{\mathrm{Aff}}^{\psi} \, \Big| \, \psi \in L^{2}(\mathbb{R}_+) \right\} \subset L_{r}^{2}(\mathrm{Aff}),\] and call $\mathfrak{W}(\mathrm{Aff})$ the \textit{affine Wigner space}. The affine orthogonality relation \eqref{affine_orthogonality_relation_new} implies that $\mathfrak{W}(\mathrm{Aff})$ is a closed subset of $L_{r}^{2}(\mathrm{Aff})$. Although we can create orthonormal bases for $L_{r}^{2}(\mathrm{Aff})$ by using the affine cross-Wigner transform as done in Lemma \ref{orthogonality_Laguerre}, the space $\mathfrak{W}(\mathrm{Aff})$ is a proper subset of $L_{r}^{2}(\mathrm{Aff})$. \par 
Despite the fact that an arbitrary function $f \in L_{r}^{2}(\mathrm{Aff})$ is not in the affine Wigner space, it is natural to ask how far $f$ is from being in $\mathfrak{W}(\mathrm{Aff})$. Hence we are interested in the following \textit{affine Wigner approximation problem}: Given a function $f \in L_{r}^{2}(\mathrm{Aff})$, we want to understand the quantity 
\begin{equation}
\label{minimization_problem}
    \inf_{g \in \mathfrak{W}(\mathrm{Aff})}\|f-g\|_{L_{r}^{2}(\mathrm{Aff})}.
\end{equation}
At this point, it should be clear to the reader that the norm on $L_{r}^{2}(\mathrm{Aff})$ is the most natural choice to consider. The analogous problem for the classical Wigner distribution has been recently investigated in \cite{ben2018wigner}. Our approach will be different from the one taken in \cite{ben2018wigner} as it will emphasize the quantization picture. \par 
Let us begin by discussing what we might expect to obtain: Our goal is to understand \eqref{minimization_problem} in terms of intrinsic properties of the function $f$. To be more precise, consider $g = W_{\mathrm{Aff}}^{\psi}$ for some $\psi \in L^{2}(\mathbb{R}_+)$. Then \eqref{affine_weyl_correspondence} and the affine orthogonality relation show that \[\left\langle A_{g}\phi, \phi \right\rangle = \left\langle g, W_{\mathrm{Aff}}^{\phi} \right\rangle = \left\langle W_{\mathrm{Aff}}^{\psi}, W_{\mathrm{Aff}}^{\phi} \right\rangle = |\langle \psi,\phi\rangle|^2,\] for $\phi \in L^{2}(\mathbb{R}_{+}).$ It follows that the 
Hilbert-Schmidt operator $A_{g}$ on $L^{2}(\mathbb{R}_{+})$ is the positive rank-one operator 
\begin{equation}
\label{positive_rank_one_operators}
    A_{g}\phi =\langle \phi, \psi \rangle \psi.
\end{equation}
The converse follows as well, so there is a one-to-one correspondence between affine Wigner distributions and the positive rank-one operators given in \eqref{positive_rank_one_operators}. Hence the distance \eqref{minimization_problem} should somehow be related to how far $A_f$ is from being a rank-one operator. We will show in Corollary \ref{Corollary_approximation} that, for a large class of functions $f \in L_{r}^{2}(\mathrm{Aff})$, this heuristic is correct. We use the notation \[\lambda_{\max}^{+}(A_f) := \max\left\{\max_{\lambda \in \mathrm{Spec}(A_{f})}\lambda,0\right\}.\] If $A_f$ is a positive operator, then $\lambda_{\max}^{+}(A_f)$ coincides with the spectral radius of $A_f$.

\begin{theorem}
\label{minimizing_theorem}
The affine Wigner approximation problem for a real-valued function $f \in L_{r}^{2}(\mathrm{Aff})$ has the explicit solution 
\begin{equation}
\label{approximation_statement}
    \inf_{g \in \mathfrak{W}(\mathrm{Aff})}\|f-g\|_{L_{r}^{2}(\mathrm{Aff})} = \sqrt{\|f\|_{L_{r}^{2}(\mathrm{Aff})}^2 - \lambda_{\max}^{+}(A_f)^{2}}.
\end{equation}
A minimizing function $h \in L_{r}^{2}(\mathrm{Aff})$ such that \[\inf_{g \in \mathfrak{W}(\mathrm{Aff})}\|f-g\|_{L_{r}^{2}(\mathrm{Aff})} = \|f-h\|_{L_{r}^{2}(\mathrm{Aff})}\] always exists. Moreover, when $\lambda_{\max}^{+}(A_f) > 0$ the number of minimizers is equal to the multiplicity of $\lambda_{\max}^{+}(A_f)$. If $\lambda_{\max}^{+}(A_f) = 0$, then the zero function is the unique minimizer.
\end{theorem}

\begin{proof}
From Proposition \ref{quantization_isometry_proposition} and the discussion above, we have that \[\inf_{g \in \mathfrak{W}(\mathrm{Aff})}\|f-g\|_{L_{r}^{2}(\mathrm{Aff})} = \inf_{\psi \in L^{2}(\mathbb{R}_+)}\|A_f - \psi \otimes \overline{\psi}\|_{\mathcal{HS}}.\] Since $A_f$ is a Hilbert-Schmidt operator it is in particular a compact operator. Moreover, $A_f$ is self-adjoint since 
\[\left\langle A_{f}\psi, \phi \right\rangle =  \left\langle A_{\overline{f}}\psi, \phi \right\rangle = \left\langle \overline{f}, W_{\mathrm{Aff}}^{\phi,\psi} \right\rangle = \overline{\left\langle f, W_{\mathrm{Aff}}^{\psi,\phi} \right\rangle} = \left\langle \psi, A_{f}\phi \right\rangle,\]
for $\psi,\phi \in L^{2}(\mathbb{R}_+).$ Thus the spectral theory for compact, self-adjoint operators implies that the spectrum $\mathrm{Spec}(A_f) = \{\lambda_{k}\}_{k = 0}^{\infty}$ of $A_f$ is countable with $0 \in \mathrm{Spec}(A_f)$ as the only possible accumulation point. Moreover, there is by \cite[Theorem 1.52]{folland2016course} an orthonormal basis $\{\phi_k\}_{k = 0}^{\infty}$ for $L^{2}(\mathbb{R}_+)$ such that $\phi_k$ is an eigenvector $A_f$ corresponding to the eigenvalue $\lambda_k$. The convention that eigenvalues with higher multiplicity than one are repeated according to their multiplicity is used. \par 
We claim that we can write $A_f = \sum_{k = 0}^{\infty} \lambda_k \phi_k \otimes \overline{\phi_k}$, where the convergence is in the Hilbert-Schmidt norm. Notice that convergence of $\sum_{k = 0}^{\infty} \lambda_k \phi_k \otimes \overline{\phi_k}$ to $A_{f}$ is guaranteed in the operator norm from the theory of compact operators \cite[Theorem 3.5]{busch2016quantum}. Hence it suffices to show that $\sum_{k = 0}^{\infty} \lambda_k \phi_k \otimes \overline{\phi_k}$ converges in the Hilbert-Schmidt norm; this will imply together with the norm inequality $\|\cdot\|_{op} \leq \|\cdot\|_{\mathcal{HS}}$ that $\sum_{k = 0}^{\infty} \lambda_k \phi_k \otimes \overline{\phi_k}$ must converge to $A_f$ in the Hilbert-Schmidt norm. Since $\mathcal{S}_{2}(L^{2}(\mathbb{R}_{+}))$ is complete, it suffices to show that $\sum_{k = 0}^{\infty} \lambda_k \phi_k \otimes \overline{\phi_k}$ is a Cauchy sequence. For $n,m \in \mathbb{N}$ with $n < m$ we have 
\begin{align*}
    \left\|\sum_{k = n}^{m} \lambda_k \phi_k \otimes \overline{\phi_k} \right\|_{\mathcal{HS}}^2 & = \left \langle \sum_{k = n}^{m} \lambda_k \phi_k \otimes \overline{\phi_k}, \sum_{k' = n}^{m} \lambda_{k'} \phi_{k'} \otimes \overline{\phi_{k'}} \right \rangle_{\mathcal{HS}} \\ & = \sum_{k,k' = n}^{m}\lambda_{k} \overline{\lambda_{k'}}  \left \langle \phi_k \otimes \overline{\phi_k}, \phi_{k'} \otimes \overline{\phi_{k'}} \right \rangle_{\mathcal{HS}} \\ & = \sum_{k = n}^{m} |\lambda_{k}|^2.
\end{align*}
The claim follows from the fact that $A_f$ is a Hilbert-Schmidt operator. \par 
Returning to the problem, we can now write 
\begin{equation}
\label{equation_in_proof_minimized}
    \inf_{g \in \mathfrak{W}(\mathrm{Aff})}\|f-g\|_{L_{r}^{2}(\mathrm{Aff})} = \inf_{\psi \in L^{2}(\mathbb{R}_+)}\left\|\sum_{k = 0}^{\infty} \lambda_k \phi_k \otimes \overline{\phi_k} - \psi \otimes \overline{\psi}\right\|_{\mathcal{HS}}.
\end{equation}
Assume that $\lambda_j = \lambda_{\max}^{+}(A_f)$. Then \eqref{equation_in_proof_minimized} is clearly minimized when $\psi = \sqrt{\lambda_j} \phi_j$. By orthogonality, we can rewrite \eqref{equation_in_proof_minimized} and obtain 
\[\inf_{g \in \mathfrak{W}(\mathrm{Aff})}\|f-g\|_{L_{r}^{2}(\mathrm{Aff})} = \sqrt{\|A_{f}\|_{\mathcal{HS}}^2 - \lambda_{\max}^{+}(A_f)^{2}} = \sqrt{\|f\|_{L_{r}^{2}(\mathrm{Aff})}^2 - \lambda_{\max}^{+}(A_f)^{2}}.\]
We always have a minimizer as we can take $h = W_{\mathrm{Aff}}^{\psi}$. The statement about uniqueness of minimizers is clear from \eqref{equation_in_proof_minimized}.
\end{proof}

\begin{remarks}
\begin{itemize}
    \item From the spectral theory of compact, self-adjoint operators, it also follows that the eigenspaces corresponding to non-zero eigenvalues are finite-dimensional. Hence, for a given $f \in L_{r}^{2}(\mathrm{Aff})$, there is at most a finite number of minimizers $h_1, \dots, h_k \in L_{r}^{2}(\mathrm{Aff})$ so that \[\inf_{g \in \mathfrak{W}(\mathrm{Aff})}\|f-g\|_{L_{r}^{2}(\mathrm{Aff})} = \|f-h_i\|_{L_{r}^{2}(\mathrm{Aff})},\] for $i = 1, \dots, k$. The statements in Theorem \ref{minimization_problem} regarding existence and uniqueness of minimizer do not follow automatically from Hilbert space theory as the affine Wigner space $\mathfrak{W}(\mathrm{Aff})$ is not convex. 
    \item The proof of Theorem \ref{minimization_problem} goes through almost verbatim to show the analogous result for the classical Wigner distribution. The analogous formula to \eqref{approximation_statement} for the classical Wigner distribution was shown in \cite[Theorem 3]{ben2018wigner} using a variational calculus approach. 
\end{itemize}
\end{remarks}

\begin{corollary}
\label{Corollary_approximation}
Let $f \in L_{r}^{2}(\mathrm{Aff})$ be real-valued and assume that \[\lambda_{\max}^{+}(A_f) = \max_{\lambda \in \mathrm{Spec}(A_f)} |\lambda|.\] Then \begin{equation}
\label{corollary_equation_approximation}
    \min_{g \in \mathfrak{W}(\mathrm{Aff})}\|f-g\|_{L_{r}^{2}(\mathrm{Aff})} = \sqrt{\|A_f\|_{\mathcal{HS}}^2 - \|A_{f}\|_{op}^2}.
\end{equation}
\end{corollary}

\begin{proof}
Since $A_f$ is self-adjoint it follows from \cite[Proposition 1.24]{folland2016course} that \[\|A_f\|_{op} = \max_{\lambda \in \mathrm{Spec}(A_f)} |\lambda|.\] Hence the result follows from Theorem \ref{minimizing_theorem}.
\end{proof}

\begin{remark}
Notice that under the assumptions in Corollary \ref{Corollary_approximation}, the heuristic we presented regarding rank-one operators holds true: If $A_f$ is a rank-one operator, then the Hilbert-Schmidt norm and the operator norm coincide. Hence \eqref{corollary_equation_approximation} is zero and thus $f$ is in the affine Wigner space $\mathfrak{W}(\mathrm{Aff})$. Conversely, the equations 
\begin{equation}
\label{eigenvalues_explicit}
    \|A_f\|_{op}^2 = \max_{\lambda \in \mathrm{Spec}(A_f)} \lambda^2, \qquad \|A_f\|_{\mathcal{HS}}^2 = \sum_{\lambda \in \mathrm{Spec}(A_f)}\lambda^2
\end{equation}
imply that \eqref{corollary_equation_approximation} is zero precisely when $A_f$ is a rank-one operator. 
\end{remark}

\begin{example}
Let $f \in L_{r}^{2}(\mathrm{Aff})$ be such that $A_{f}$ is a positive operator with rank $k > 0$. Then \eqref{eigenvalues_explicit} implies that \[\|A_{f}\|_{op}^{2} \geq \frac{\|A_{f}\|_{\mathcal{HS}}^2}{k}.\] Hence we obtain from \eqref{corollary_equation_approximation} that \[\min_{g \in \mathfrak{W}(\mathrm{Aff})}\|f-g\|_{L_{r}^{2}(\mathrm{Aff})} = \sqrt{\|A_f\|_{\mathcal{HS}}^2 - \|A_{f}\|_{op}^2} \leq \sqrt{\frac{k-1}{k}}\|A_{f}\|_{\mathcal{HS}} = \sqrt{\frac{k-1}{k}}\|f\|_{L_{r}^{2}(\mathrm{Aff})}.\] This has the following consequence: Let $f_1, f_2 \in L_{r}^{2}(\mathrm{Aff})$ both correspond to positive operators $A_{f_1}$ and $A_{f_2}$ with finite rank. If $\mathrm{rank}(A_{f_1}) \ll \mathrm{rank}(A_{f_2})$, then $f_1$ will be closer to the affine Wigner space than $f_2$, unless the energy of $f_2$ is significantly smaller than that of $f_1$.
\end{example}

\subsection{Applications to Operators on $\mathbb{R}_{+}$}

We give two minor applications to operators on $\mathbb{R}_{+}$. 
An operator $A \in \mathcal{B}(L^{2}(\mathbb{R}_+))$ is said to be \textit{dilation invariant} if 
\begin{equation}
\label{dilation_invariant}
    A = D_{r}\circ A \circ D_{r}^{*},
\end{equation}
for all $r > 0$ where $D_r$ is the dilation operator in \eqref{dilation_operator}. We will use the affine quantization scheme to show the following result.

\begin{proposition}
\label{no_dilation_invariant_operators}
There are no non-zero dilation invariant Hilbert-Schmidt operators on $L^{2}(\mathbb{R}_+)$.
\end{proposition}

\begin{proof}
Assume by contradiction that $A \in \mathcal{S}_{2}(L^{2}(\mathbb{R}_+))$ is dilation invariant. The quantization correspondence implies that $A$ is on the form $A = A_f$ for some $f \in L_{r}^{2}(\mathrm{Aff})$. It follows from \eqref{elementary_symmetries_of_affine_Wigner} that \[W_{\mathrm{Aff}}^{D_{\frac{1}{r}}\psi, D_{\frac{1}{r}}\phi}(x,a) = r\cdot W_{\mathrm{Aff}}^{\psi,\phi}(x,ra), \qquad \psi,\phi \in L^{2}(\mathbb{R}_+),\] for $r > 0$ and $(x,a) \in \mathrm{Aff}$. Hence \eqref{affine_weyl_correspondence} implies that
\begin{align*}
    \left\langle D_{r} A_{f} D_{\frac{1}{r}}\psi, \phi \right\rangle & = \left\langle A_{f} D_{\frac{1}{r}}\psi, D_{\frac{1}{r}} \phi \right\rangle \\ & = \left\langle f, W_{\mathrm{Aff}}^{D_{\frac{1}{r}}\phi, D_{\frac{1}{r}}\psi} \right\rangle \\ & = \int_{-\infty}^{\infty}\int_{0}^{\infty} rf(x,a) \overline{W_{\mathrm{Aff}}^{\phi,\psi}(x,ra)} \, \frac{da \, dx}{a} \\ & = \int_{-\infty}^{\infty}\int_{0}^{\infty} rf\left(x,\frac{a}{r}\right) W_{\mathrm{Aff}}^{\psi,\phi}(x,a) \, \frac{da \, dx}{a}.
\end{align*}
On the other hand, since $A_f$ is dilation invariant we also have \[\left\langle D_{r} A_{f} D_{\frac{1}{r}}\psi, \phi \right\rangle = \left\langle A_{f} \psi, \phi \right\rangle = \left\langle f, W_{\mathrm{Aff}}^{\phi, \psi} \right\rangle = \int_{-\infty}^{\infty}\int_{0}^{\infty} f\left(x,a\right) W_{\mathrm{Aff}}^{\psi,\phi}(x,a) \, \frac{da \, dx}{a}.\]By Lemma \ref{Wigner_distribution_bases}, this forces $f \in L_{r}^{2}(\mathrm{Aff})$ to satisfy the homogeneity relation \[f(x,a) = rf\left(x,\frac{a}{r}\right),\] for all $r \in \mathbb{R}_{+}$ and almost every $(x,a) \in \mathrm{Aff}$. However, this implies that \[\|f\|_{L_{r}^{2}(\mathrm{Aff})}^2 = \int_{-\infty}^{\infty} \int_{0}^{\infty} |f(x,a)|^2 \, \frac{da \, dx}{a} = r^2 \int_{-\infty}^{\infty} \int_{0}^{\infty} \left|f\left(x,\frac{a}{r}\right)\right|^2 \, \frac{da \, dx}{a} = r^2\|f\|_{L_{r}^{2}(\mathrm{Aff})}^2.\]
Hence $f$ is not in $L_{r}^{2}(\mathrm{Aff})$ unless $f = 0$, in which case $A_f$ is the zero operator.
\end{proof}

\begin{remarks}
\begin{itemize}
\item Notice that the proof of Proposition \ref{no_dilation_invariant_operators} actually shows that there can be no non-zero Hilbert-Schmidt operator $A$ that satisfies \eqref{dilation_invariant} even for a single $r \neq 1$. In particular, there are no non-zero Hilbert-Schmidt operators on $L^{2}(\mathbb{R}_+)$ that are \textit{discretely dilation invariant} in the sense that \[A = D_{2^n} \circ A \circ D_{\frac{1}{2^n}}, \quad n \in \mathbb{Z}.\]
\item Consider the space $\mathcal{M}_{(0,\infty)}$ of all $\psi \in L^{2}(\mathbb{R}_+)$ such that the Mellin transform of $\psi$ satisfies \[\mathrm{supp}(\mathcal{M}(\psi)) \subset \mathbb{R}_{+}.\] Then the orthogonal projection $P:L^{2}(\mathbb{R}_+) \to \mathcal{M}_{(0,\infty)}$ is dilation invariant due to the dilation property of the Mellin transform given in \eqref{Mellin_dilation}. Hence there are non-trivial dilation invariant operators on $L^{2}(\mathbb{R}_+)$.
\end{itemize}
\end{remarks}

We end by giving an application to the trace class operators on $L^{2}(\mathbb{R}_{+})$. The following result is motivated by \cite[Proposition 162]{de2017wigner}.

\begin{corollary}
\label{trace_class_proposition}
Let $T \in \mathcal{S}_{1}(L^{2}(\mathbb{R}_{+}))$ be a trace class operator. Then we can write $T = A_{f} \circ A_{g}$ for $f,g \in L_{r}^{2}(\mathrm{Aff})$. Moreover, the trace of $T$ can be calculated by the formula \[\mathrm{Tr}(T) = \mathrm{Tr}(A_{f} \circ A_{g}) = \int_{-\infty}^{\infty}\int_{0}^{\infty}f(x,a)g(x,a) \, \frac{da \, dx}{a}.\]
\end{corollary}

\begin{proof}
As mentioned in Appendix \ref{Appendix_1}, any trace class operator $T$ on $L^{2}(\mathbb{R}_{+})$ can be written as a composition of Hilbert-Schmidt operators $T = A \circ B$ with $A,B \in \mathcal{S}_{2}(L^{2}(\mathbb{R}_{+}))$. We can now use the bijective correspondence between Hilbert-Schmidt operators on $L^{2}(\mathbb{R}_{+})$ and $L_{r}^{2}(\mathrm{Aff})$ to write $A = A_f$ and $B = A_g$ for $f,g \in L_{r}^{2}(\mathrm{Aff})$. Thus by \eqref{composition_of_HS_operators} we can write 
\begin{equation*}
    \mathrm{Tr}(T) = \mathrm{Tr}(A_{f} \circ A_{g}) = \left\langle A_{g}, A_{f}^{*} \right\rangle_{\mathcal{HS}} = \left\langle g, \overline{f} \right\rangle_{L_{r}^{2}(\mathrm{Aff)}} = \int_{-\infty}^{\infty}\int_{0}^{\infty}f(x,a)g(x,a) \, \frac{da \, dx}{a}.
\end{equation*}

\vspace{-1cm}
   
\end{proof}
\vspace{0.2cm}

\begin{remark}
Notice that \[\mathrm{Tr}\left(T^{*}\right) = \mathrm{Tr}\left(A_{g}^{*} \circ A_{f}^{*}\right) = \mathrm{Tr}\left(A_{\overline{g}} \circ A_{\overline{f}}\right) = \int_{-\infty}^{\infty}\int_{0}^{\infty}\overline{f(x,a)g(x,a)} \, \frac{da \, dx}{a}.\] This can also be deduced from \eqref{trace_of_adjoint}. In particular, the trace of $T$ is real-valued whenever $f$ and $g$ are real-valued.
\end{remark}

\section{Further Research}
\label{sec: Open_Questions}

\subsubsection*{Affine Grossmann-Royer Operator}
A standard tool for deriving properties of the classical Wigner distribution is the \textit{Grossmann-Royer operator} $\widehat{R}(x,\omega)$ defined by the relation \[W(f,g)(x,\omega) = \left\langle \widehat{R}(x,\omega)f, g \right\rangle_{L^{2}(\mathbb{R}^d)},\] for $f,g \in L^{2}(\mathbb{R}^d)$ and $(x,\omega) \in \mathbb{R}^{2d}$. The precise formula for $\widehat{R}(x,\omega)$ can be found in \cite[Chapter 1]{de2017wigner} with a different normalization. An essential property of the Grossmann-Royer operator $\widehat{R}(x,\omega)$ is that \[\left\|\widehat{R}(x,\omega)f\right\|_{L^{2}(\mathbb{R}^{d})} = 2^d \cdot \|f\|_{L^{2}(\mathbb{R}^d)}, \] for all $f \in L^{2}(\mathbb{R}^d)$ and $(x,\omega) \in \mathbb{R}^{2d}$. This is immensely useful; to see that the classical cross-Wigner transform is bounded one simply needs to apply Cauchy-Schwarz inequality to obtain 
\begin{equation}
\label{Grossmann_bounded_equation}
    \sup_{(x,\omega) \in \mathbb{R}^{2d}}|W(f,g)(x,\omega)| \leq \left\|\widehat{R}(x,\omega)f\right\|_{L^{2}(\mathbb{R}^d)} \|g\|_{L^{2}(\mathbb{R}^d)} = 2^d \cdot \|f\|_{L^{2}(\mathbb{R}^d)} \|g\|_{L^{2}(\mathbb{R}^d)}.
\end{equation} \par
Analogously, we define the \textit{affine Grossmann-Royer operator} $\widehat{R}_{\mathrm{Aff}}(x,a)$ by the relation \[W_{\mathrm{Aff}}^{\psi,\phi}(x,a) = \left\langle \widehat{R}_{\mathrm{Aff}}(x,a)\psi, \phi \right\rangle_{L^{2}(\mathbb{R}_+)},\] for $\psi,\phi \in \mathcal{S}(\mathbb{R}_+)$ and $(x,a) \in \mathrm{Aff}$. We restrict our attention to Schwartz functions for convenience since then $W_{\mathrm{Aff}}^{\psi,\phi} \in \mathcal{S}(\mathrm{Aff})$ and hence have well-defined point values. Notice that the affine Grossmann-Royer operator $\widehat{R}_{\mathrm{Aff}}(x,a)$ is precisely the affine quantization of the point mass $\delta_{\mathrm{Aff}}(x,a)$ given in Example \ref{point_mass_example}.
\begin{lemma}
The affine Grossmann-Royer operator have the explicit form \[\widehat{R}_{\mathrm{Aff}}(x,a)\psi(r) = \frac{e^{2 \pi i x \lambda^{-1}\left(\frac{r}{a}\right)}\lambda^{-1}\left(\frac{r}{a}\right)\left(1 - e^{\lambda^{-1}\left(\frac{r}{a}\right)}\right)}{1 + \lambda^{-1}\left(\frac{r}{a}\right) - e^{\lambda^{-1}\left(\frac{r}{a}\right)}} \cdot \psi\left(r e^{-\lambda^{-1}\left(\frac{r}{a}\right)}\right),\]
for $\psi \in \mathcal{S}(\mathbb{R}_+)$, $r > 0$, and $(x,a) \in \mathrm{Aff}$ where $\lambda$ is the function given in \eqref{lambda_function}.
\end{lemma}
Trying to generalize the strategy in \eqref{Grossmann_bounded_equation} runs into a problem: The affine Grossmann-Royer operator is not a bounded operator on $\mathcal{S}(\mathbb{R}_+) \subset L^{2}(\mathbb{R}_{+})$ with respect to the norm $\|\cdot\|_{L^{2}(\mathbb{R}_+)}$. However, if $\psi \in \mathcal{S}(\mathbb{R}_{+})$ is supported in the interval $\left[\frac{1}{k}, k\right]$ for some $k > 0$, then there is a constant $C_k > 0$ such that \[\left\|\widehat{R}_{\mathrm{Aff}}(x,a)\psi\right\|_{L^{2}(\mathbb{R}_{+})} \leq C_k \cdot \|\psi\|_{L^{2}(\mathbb{R}_{+})}.\] We call the optimal constant $C_k$ in the inequality above the \textit{$k$-support constant}. Hence if $\phi \in \mathcal{S}(\mathbb{R}_{+})$ we have \[\sup_{(x,a) \in \mathrm{Aff}}\left|W_{\mathrm{Aff}}^{\psi,\phi}(x,a)\right| \leq C_k \cdot \|\psi\|_{L^{2}(\mathbb{R}_+)}\|\phi\|_{L^{2}(\mathbb{R}_+)}.\]
A trivial adaption of \cite[Lemma 4.3.7]{grochenig2013foundations} gives the following \textit{relative uncertainty principle} for the affine Wigner distribution.

\begin{proposition}
\label{Wigner_uncertainty_prop}
Let $\psi \in \mathcal{S}(\mathbb{R}_{+})$ be supported in the interval $\left[\frac{1}{k}, k\right]$ for some $k > 0$ and let $U \subset \mathrm{Aff}$ be a Borel set. Assume there is an $\epsilon \geq 0$ such that \[\int_{U} W_{\mathrm{Aff}}^{\psi}(x,a) \, \frac{da \, dx}{a} \geq (1 - \epsilon)\|\psi\|_{L^{2}(\mathbb{R}_+)}^{2}.\] Then the right Haar measure of $U$ satisfies \[\mu_{r}(U) \geq \frac{1 - \epsilon}{C_k},\]
where $C_k$ is the $k$-support constant.
\end{proposition}
Motivated by Proposition \ref{Wigner_uncertainty_prop}, it is of interest to investigate the $k$-support constant $C_k$ both numerically and asymptotically.

\subsubsection*{Affine Positivity Conjecture}

One of the major results about the classical Wigner distribution is regarding positivity; when is $W_f$ a non-negative function on $\mathbb{R}^{2d}$? Normalized functions $f \in L^{2}(\mathbb{R}^{d})$ such that $W_f$ is non-negative would generate probability density functions on $\mathbb{R}^{2d}$ that represent the time-frequency distribution of $f$. However, a well-known result of Hudson \cite[Theorem 4.4.1]{grochenig2013foundations} shows that this can only happen for suitably perturbed Gaussians. \par Turning to the affine setting, we would like to determine the normalized functions $\psi \in L^{2}(\mathbb{R}_{+})$ such that $W_{\mathrm{Aff}}^{\psi}$ is a non-negative function on the affine group. In \cite{molnar2001coherent} the authors showed that the affine Wigner distribution $W_{\mathrm{Aff}}^{\psi_{s}}$ is non-negative if $\psi_{s}$ is the so called \textit{Morse ground state} 
\begin{equation*}
    \psi_{s}(r) := \frac{r^s e^{-\frac{r}{2}}}{\Gamma(2s)}, \qquad s \geq 0.
\end{equation*} 
In our case, we will only consider $\psi_{s}$ for $s > 0$ as $\psi_{0} \not \in L^{2}(\mathbb{R}_{+})$. More generally, one can use the invariance properties \eqref{invariance_property} and \eqref{hyperbolic_invariance} to show that the affine Wigner distribution $W_{\mathrm{Aff}}^{\psi}$ of 
\begin{equation}
\label{almost_ Klauder}
    \psi(r) = Cr^{-i(x + i a)}e^{i(y + ib)r},
\end{equation}
is non-negative when $C \in \mathbb{C}$ and $(x,a), (y,b) \in \mathrm{Aff}$. The functions on the form \eqref{almost_ Klauder} are the \textit{generalized Klauder wavelets} in \cite[Equation (41)]{flandrin1998separability} that are in $L^{2}(\mathbb{R}_+)$. This leads to the following affine positivity conjecture: 
\begin{center}
    The only functions $\psi \in L^{2}(\mathbb{R}_{+})$ such that $W_{\mathrm{Aff}}^{\psi}$ is non-negative on the affine group are the generalized Klauder wavelets in \eqref{almost_ Klauder}.
\end{center}

\section{Appendices}

\subsection{Proof of Proposition \ref{polyanalytic_decomposition}}
\label{Appendix_Proof}

Notice first that 
\begin{align*}
L_{r}^{2}(\mathrm{Aff}) & \simeq L^{2}(\mathbb{R},dx) \otimes L^{2}(\mathbb{R}_{+},a^{-1} \, da) \\ & \simeq \left(L^{2}(\mathbb{R}_{-},dx) \otimes L^{2}(\mathbb{R}_{+},a^{-1} \, da)\right) \oplus \left(L^{2}(\mathbb{R}_{+},dx) \otimes L^{2}(\mathbb{R}_{+},a^{-1} \, da)\right).    
\end{align*}
Hence the decomposition formulas for $\mathbb{A}^{n}(\mathrm{Aff})$ and $\mathbb{A}^{\perp,n}(\mathrm{Aff})$ given in \eqref{polyanalytic_decomposition} imply the desired orthogonal decomposition of $L_{r}^{2}(\mathrm{Aff})$ because the functions $\mathcal{L}_{n}$ form an orthonormal basis for $L^{2}(\mathbb{R}_{+},a^{-1} \, da)$. Moreover, the decompositions of the pure spaces $\mathcal{A}^{n}(\mathrm{Aff})$ and $\mathcal{A}^{\perp,n}(\mathrm{Aff})$ also follow from the decompositions of $\mathbb{A}^{n}(\mathrm{Aff})$ and $\mathbb{A}^{\perp,n}(\mathrm{Aff})$. We will focus on showing the decomposition 
\begin{equation}
\label{decomposition_equation_in_proof}
\mathbb{A}^{n}(\mathrm{Aff}) \simeq L^{2}(\mathbb{R}_{+},dx) \otimes \bigoplus_{k = 0}^{n-2}\mathrm{span}\left\{\mathcal{L}_{k}\right\},  
\end{equation}
since the decomposition of $\mathbb{A}^{\perp,n}(\mathrm{Aff})$ is similar. \par The idea of the proof is to define several isometries of the space $L_{r}^{2}(\mathrm{Aff})$ so that the equation $\partial_{\overline{z}}f = 0$ is transformed into something more manageable.
Define the two isometries \[\mathcal{F} \otimes I:L^{2}(\mathbb{R},dx) \otimes L^{2}(\mathbb{R}_{+},da) \longrightarrow L^{2}(\mathbb{R},dx) \otimes L^{2}(\mathbb{R}_{+},da),\] \[\mathcal{V}:L^{2}(\mathbb{R},dx) \otimes L^{2}(\mathbb{R}_{+},da) \longrightarrow L^{2}(\mathbb{R},dx) \otimes L^{2}(\mathbb{R}_{+},da),\] where $\mathcal{F}$ denotes the Fourier transform and $\mathcal{V}$ is given by \[\mathcal{V}(f)(x,a) = \frac{1}{\sqrt{2|x|}}f\left(x,\frac{a}{2|x|}\right).\]
Let $M_{g}$ denote the multiplication operator with symbol $g$. We consider the following diagram
\[
\begin{tikzcd}
[row sep=scriptsize, column sep=scriptsize]
& L^{2}(\mathrm{Aff},dx \, da) \simeq L^{2}(\mathbb{R},dx) \otimes L^{2}(\mathbb{R}_{+},da) \arrow[rr,"\mathcal{V} \circ \mathcal{F} \otimes I"] & & L^{2}(\mathbb{R},dx) \otimes L^{2}(\mathbb{R}_{+},da) \simeq L^{2}(\mathrm{Aff},dx \, da)\arrow[dd, "M_{\sqrt{a}}"]\\ \\
& L_{r}^{2}(\mathrm{Aff}) \simeq L^{2}(\mathbb{R},dx) \otimes L^{2}(\mathbb{R}_{+},a^{-1} \, da) \arrow[rr,"\Phi"] \arrow[uu, "M_{\frac{1}{\sqrt{a}}}"] & & L^{2}(\mathbb{R},dx) \otimes L^{2}(\mathbb{R}_{+},a^{-1} \, da) \simeq L_{r}^{2}(\mathrm{Aff})
\end{tikzcd}\] where $\Phi$ is defined to be the composition of the three other isometries. The image of $\mathbb{A}^{n}(\mathrm{Aff})$ under $\Phi$ consists of all functions in $L_{r}^{2}(\mathrm{Aff})$ that satisfies
    \begin{equation}
    \label{equation_transformed}
        \Phi(\partial_{x} + i \partial_{a})^{n}\Phi^{-1}f = 0.
    \end{equation}
    Since the Fourier transform is only applied to the second variable, we have \[\Phi f(x,a) = M_{\sqrt{a}}\mathcal{V}(\mathcal{F} \otimes I)M_{\frac{1}{\sqrt{a}}}f(x,a) = M_{\sqrt{2|x|}}\mathcal{V}(\mathcal{F} \otimes I)f(x,a).\] \par 
    When computing how the equation $\partial_{\overline{z}}^{n}f = 0$ gets transformed under $\Phi$, we remove constants as we go along as they do not change the solution to a homogeneous differential equation. Expanding the right hand side of (\ref{equation_transformed}) shows that
    \begin{align*}
 \Phi(\partial_{x} + i \partial_{a})^{n}\Phi^{-1}f(x,a) = |x|^n\left(\mathrm{sign}(x) + 2 \partial_a\right)^{n}f(x,a).
    \end{align*}
    Hence the solutions to \eqref{equation_transformed} are on the form \begin{equation}
    \label{transformed_equation_2}
        \sum_{k = 1}^{n-1}\mathbf{1}_{[0,\infty)}(x)\psi_{k}(x)a^{k}e^{-\frac{a}{2}}, \qquad \psi_k \in L^{2}(\mathbb{R},dx).
    \end{equation} \par 
    Notice that when $n = 1$ there are no solutions to \eqref{equation_transformed}, reflecting Proposition \ref{no_analytic_functions}. We can rewrite \eqref{transformed_equation_2}
    as all functions $f \in L_{r}^{2}(\mathrm{Aff})$ on the form \[f(x,a) = \sum_{k = 0}^{n-2}g_{k}(x)\mathcal{L}_{k}(a), \qquad g_k \in L^{2}(\mathbb{R}_+,dx).\] Thus we obtain the decomposition \eqref{decomposition_equation_in_proof} and the rest of the result follows from the remarks made in the beginning of the proof.

\subsection{Schatten Class Operators}
\label{Appendix_1}
As we need Hilbert-Schmidt operators throughout the paper and trace class operators in Corollary \ref{trace_class_proposition}, we review basic facts about Schatten class operators $\mathcal{S}_{p}(\mathcal{H})$ on an arbitrary separable Hilbert space $\mathcal{H}$. In this framework, the trace class operators are simply $\mathcal{S}_{1}(\mathcal{H})$, while the Hilbert-Schmidt operators are $\mathcal{S}_{2}(\mathcal{H})$. A thorough reference for Schatten class operators is \cite[Chapter 3]{busch2016quantum}. \par 
Recall that the \textit{singular values} $\{s_{n}(T)\}_{n = 0}^{\infty}$ of a bounded operator $T$ on a separable Hilbert space $\mathcal{H}$ are the eigenvalues of the operator $|T| := \sqrt{T^* T}$. The Schatten classes consist of operators where the singular values decay suitably fast: Define the $p$'th \textit{Schatten class operators} $\mathcal{S}_{p}(\mathcal{H})$ on a separable Hilbert space $\mathcal{H}$ for $1 \leq p < \infty$ to be all bounded operators $T$ such that the singular values $\{s_{n}(T)\}_{n = 0}^{\infty}$ of $T$ satisfies
\begin{equation}
\label{Schatten norm}
    \|T\|_{\mathcal{S}_{p}(\mathcal{H})} := \left(\sum_{n = 0}^{\infty}s_{n}(T)^{p}\right)^{\frac{1}{p}} < \infty.
\end{equation}

The expression \eqref{Schatten norm} makes $\mathcal{S}_{p}(\mathcal{H})$ into a Banach space for all $1 \leq p < \infty$ and the Schatten class operators are compact operators for all $1 \leq p < \infty$. Moreover, we have the norm estimates \[\|T\|_{op} \leq \|T\|_{\mathcal{S}_{q}(\mathcal{H})} \leq \|T\|_{\mathcal{S}_{p}(\mathcal{H})}, \qquad T \in \mathcal{S}_{p}(\mathcal{H}), \quad 1 \leq p \leq q < \infty.\] 

We are primarily interested in the following two cases: 
\begin{itemize}
\item For $p = 1$ we call $\mathcal{S}_{1}(\mathcal{H})$ the \textit{trace class operators}. The norm on $\mathcal{S}_{1}(\mathcal{H})$ is equivalently given by 
\begin{equation}
\label{trace_class_norm}
    \|T\|_{\mathcal{S}_{1}(\mathcal{H})} = \sum_{n = 0}^{\infty}\langle |T|\psi_n, \psi_n \rangle,
\end{equation}
where $\{\psi_n\}_{n = 0}^{\infty}$ is an orthonormal basis for $\mathcal{H}$. As expected, the expression \eqref{trace_class_norm} does not depend on the choice of orthonormal basis. Hence we can define the \textit{trace} of any trace class operator $T$ as the sum of the absolutely convergent series \[\mathrm{Tr}(T) = \sum_{n = 0}^{\infty}\langle T\psi_n, \psi_n \rangle,\] where $\{\psi_n\}_{n \in \mathbb{N}}$ is again any orthonormal basis for $\mathcal{H}$. Notice that this extends the usual trace of matrices to the setting of bounded operators on separable Hilbert spaces. The trace class operators are closed under taking adjoints, and we have the formula 
\begin{equation}
\label{trace_of_adjoint}
    \mathrm{Tr}(T^*) = \overline{\mathrm{Tr}(T)}, \qquad T \in \mathcal{S}_{1}(\mathcal{H}).
\end{equation}
\item For $p = 2$ we call $\mathcal{S}_{2}(\mathcal{H})$ the \textit{Hilbert-Schmidt operators}. This becomes a Hilbert space under the inner-product
\[\langle A,B \rangle_{\mathcal{HS}} := \sum_{n = 0}^{\infty}\langle A\psi_n,B\psi_n \rangle, \qquad A,B \in \mathcal{S}_{2}(\mathcal{H}),\]
where $\{\psi_n\}_{n = 0}^{\infty}$ is any orthonormal basis for $\mathcal{H}$. A useful fact is that any trace class operator $T \in \mathcal{S}_{1}(\mathcal{H})$ can be written as $T = A \circ B$ for two Hilbert-Schmidt operators $A,B \in \mathcal{S}_{2}(\mathcal{H})$. The trace of the operator $T$ can then be computed with the formula
\begin{equation}
\label{composition_of_HS_operators}
    \mathrm{Tr}(T) = \mathrm{Tr}(A \circ B) = \sum_{n = 0}^{\infty}\langle AB\psi_n,\psi_n \rangle = \sum_{n = 0}^{\infty}\langle B\psi_n,A^{*}\psi_n \rangle = \langle B, A^{*} \rangle_{\mathcal{HS}}.
\end{equation}
\end{itemize}


\bibliographystyle{abbrv}
\bibliography{main}

\Addresses

\end{document}